%% file: main.tex
\documentclass[conference]{IEEEtran}
\IEEEoverridecommandlockouts

\usepackage{tikz}
\usetikzlibrary{arrows.meta,bending}
\usepackage{fancyhdr}
\usepackage{amsmath,amssymb,amsfonts}
\usepackage{graphicx}
\usepackage{verbatim}
\usepackage{float}
\usepackage{caption}
\usepackage{amsmath,amsthm,amsfonts,amssymb,amscd}

\usepackage{subfigure}
\usepackage{enumerate}
\usepackage{fancyhdr}
\usepackage{mathrsfs}
\usepackage{xcolor}
\usepackage{graphicx}
\usepackage{listings}
\usepackage{float}
\usepackage{textcomp}
\usepackage{xcolor}
\usepackage{amsthm,amsmath,amssymb}
\usepackage{mathrsfs}
\usepackage{hyperref}
\usepackage{ragged2e}
\usepackage{cite}
\newcommand{\alg}{CuckooGraph}
\definecolor{fzc}{RGB}{255,0,0}

\definecolor{cyl}{RGB}{0,0,0}

\definecolor{gjr}{RGB}{0,0,0}

\newcommand{\bbb}{\noindent\textbf}
\newcommand{\xxx}{\noindent\textit}

\definecolor{reviewA}{HTML}{000000}
\definecolor{reviewB}{HTML}{000000}
\definecolor{reviewC}{HTML}{000000}

\usepackage{pifont}
\usepackage[utf8]{inputenc}
\usepackage{graphicx}
\usepackage{subfigure}
\usepackage{lipsum}
\usepackage{amsmath,bm}
\usepackage{amssymb}
\usepackage{indentfirst}
\usepackage{amsbsy}
\usepackage{cleveref}
\usepackage{mathrsfs}
\usepackage{amsthm}
\usepackage{tabularx}
\usepackage{makecell}
\usepackage{enumitem}
\usepackage[ruled, vlined, linesnumbered]{algorithm2e}
\usepackage{algpseudocode}
\usepackage{multirow} 
\usepackage[normalem]{ulem}
\makeatletter 
\usepackage{framed} 
\newtheorem{theorem}{Theorem}
\newtheorem{lemma}[theorem]{Lemma}
\newtheorem{definition}[theorem]{Definition}

\begin{document}

\title{CuckooGraph: A Scalable and Space-Time Efficient Data Structure for Large-Scale Dynamic Graphs}

\author{
\IEEEauthorblockN{
Zhuochen Fan\IEEEauthorrefmark{1}\IEEEauthorrefmark{2}, 
Yalun Cai\IEEEauthorrefmark{2}, 
Zirui Liu\IEEEauthorrefmark{2}, 
Jiarui Guo\IEEEauthorrefmark{2}, 
Xin Fan\IEEEauthorrefmark{3}, 
Tong Yang\IEEEauthorrefmark{2}, 
Bin Cui\IEEEauthorrefmark{2}
}

\IEEEauthorblockA{
\IEEEauthorrefmark{1}\small{Pengcheng Laboratory, Shenzhen, China}
\IEEEauthorrefmark{2}\small{School of Computer Science, Peking University, Beijing, China}
\IEEEauthorrefmark{3}\small{Wuhan University, Wuhan, China}
}

\texttt{\{fanzc,caiyalun,zirui.liu,ntguojiarui,yangtong,bin.cui\}@pku.edu.cn, xin.fan@whu.edu.cn}

\thanks{Co-first authors: Zhuochen Fan, Yalun Cai, and Zirui Liu.
Corresponding author: Tong Yang (yangtong@pku.edu.cn). 
}

}

\maketitle

\input{VLDB/0-abstract}

\setlength{\abovecaptionskip}{-0.00cm}
\setlength{\belowcaptionskip}{-0.02cm}
\setlength{\subfigcapskip}{-0.15cm}

\input{VLDB/1-intro}
\input{VLDB/2-rw}

\input{VLDB/3-alg}
\input{VLDB/4-math}

\input{VLDB/5-newexp}

\input{VLDB/6-con}

\section*{Acknowledgment}
    
We would like to thank the anonymous reviewers for their insightful comments.
This work was supported by the National Natural Science Foundation of China (NSFC) (No. 62402012), and in part by the China Postdoctoral Science Foundation (No. 2023TQ0010, GZC20230055, 2024M750102).

{
\bibliographystyle{IEEEtran_ref.bst}
\bibliography{references.bib}
}

\end{document}

%% file: VLDB/0-abstract.tex
\begin{abstract}

Graphs play an increasingly important role in various big data applications. 
However, existing graph data structures cannot simultaneously address the performance bottlenecks caused by the dynamic updates, large scale, and high query complexity of current graphs.
This paper proposes a novel data structure for large-scale dynamic graphs called \alg{}. 
It does not require any prior knowledge of the upcoming graphs, and can adaptively resize to the most memory-efficient form while requiring few memory accesses for very fast graph data processing.
The key techniques of \alg{} include \textsc{Transformation} and \textsc{Denylist}. 
\textsc{Transformation} fully utilizes the limited memory by designing related data structures that allow flexible space transformations to smoothly expand/tighten the required space depending on the number of incoming items.
\textsc{Denylist} efficiently handles item insertion failures and further improves processing speed.
Our experimental results show that compared with the most competitive solution Spruce, \alg{} achieves about $33\times$ higher insertion throughput while requiring only about $68\%$ of the memory space.

\end{abstract}

%% file: VLDB/1-intro.tex
\section{Introduction}
\subsection{Background and Motivation}
\label{sec:Background}

Graphs can intuitively represent various relationships between entities and are widely used in various big data applications, such as user behavior analysis in social/e-commerce networks \cite{li2017most,matsunobu2020myrocks,zhang2021group}, financial fraud detection in transactional systems \cite{wang2019semi,jiang2022spade,huang2022dgraph}, network security and monitoring in the Internet \cite{iliofotou2007network,wang2009privacy,simeonovski2017controls}, and even trajectory tracking of close contacts of the COVID-19 epidemic \cite{ma2022hierarchical,zhu2023novel}, $etc.$
Correspondingly, graph analytics systems also play an increasingly significant role, responsible for storing, processing, and analyzing graph-like data well.

As an essential part of graph analytics systems, graph storage schemes are facing challenges introduced mainly by the following properties of graph data:
\ding{172} Fast update: graphs always arrive quickly and are constantly dynamic \cite{mondal2012managing,pandey2021terrace,hou2024aeong}.
%
This requires the storage structure to be updated at high speed.
\ding{173} Large data scale: graphs can even reach hundreds of millions of edges \cite{potamias2010k,kang2012gbase,sahu2017ubiquity}.
%
%
This requires the storage structure to be flexibly adapted to the data scale.
\ding{174} High query complexity: graphs have complex topologies, and their node degrees are often unevenly distributed and follow a power-law distribution \cite{shao2014parallel,wei2019prsim,hourglasssketch}. 
This means that graphs usually consist of mostly low-degree nodes and a few high-degree nodes.
Querying the neighbors of high-degree nodes takes longer, while low-degree nodes take less time. 
However, the former is more likely to be queried and updated than the latter. 
This imbalance leads to poor query performance and hinders further optimization.
In summary, an ideal graph storage scheme can achieve memory-saving, fast processing speed, and good update and expansion performance to deal with any unknown graphs.




Currently, most existing graph storage and database schemes \cite{ediger2012stinger,shun2013ligra,zou2014gstore,sun2015sqlgraph,macko2015llama,wheatman2018packed,dhulipala2019low,firmli2020csr++,kumar2020graphone,zhu2020livegraph,feng2021risgraph,mhedhbi2021a+,islam2022vcsr,fuchs2022sortledton,qiuwind,shi2024spruce,Neo4j,OrientDB,ArangoDB,JanusGraph,GraphDB} use the following two as basic data structures: adjacency list and compressed sparse row (CSR), but neither of them directly supports large-scale dynamic graphs.
The most widely used adjacency list represents node connections intuitively and is easy to edit (such as adding or removing edges).
However, due to its non-contiguous memory allocation and inefficiency in accessing non-neighbor edges, this pointer-intensive data structure is prone to significant space-time overhead as the graph size grows.
The CSR provides a more compact array-based representation that is more memory efficient and suitable for fast traversal.
However, the CSR is inherently static and struggles with updates for dynamic graphs, as its update usually requires completely rebuilding the CSR structure, which is computationally expensive and inefficient.
In order to accommodate large-scale dynamic graphs, the data structures of many state-of-the-art graph storage schemes are evolved from adjacency lists or CSRs.
Unfortunately, they cannot completely avoid the above-mentioned shortcomings of the adjacency list or CSR itself, so that they cannot solve the above \ding{172}\ding{173}\ding{174} at the same time or have various limitations, and there is still room for improvement.
For example, Spruce \cite{shi2024spruce}, the most competitive solution whose basic data structure is based on adjacency lists, still needs to record many pointers.

\subsection{Our Proposed Solution}

In this paper, we propose a novel data structure for storing large-scale dynamic graphs, namely \textbf{\alg{}}.
It has the following advantages:
1) It is memory-saving and can be flexibly expanded or contracted according to actual operations;
2) It basically maintains the fastest running speed in a series of graph analytics tasks;
3) It works for any graph of unknown size without knowing any information about it in advance.



The design philosophy of \alg{} is as follows.
Instead of using the traditional adjacency list or CSR, we choose a hash-array-based data structure to improve time and space efficiency when handling large-scale dynamic graphs.
Specifically, we utilize a (large) cuckoo hash table (L-CHT) \cite{pagh2004cuckoo} as the basic data structure with a finer-grained partitioning of the space in each bucket.
We assume that the edge $\left \langle u, v \right \rangle$ is 
is mapped and will be stored in this bucket.
Initially, part of the bucket space is used to store node $u$, and the other part, which is divided into an even number of small slots, is used to store node(s) $v$.
Then, L-CHT decides whether to perform our \textsc{transformation} technique based on the degree (the number of incoming $v$) of the $u$:
1) When the node degree is small: this sparsity is more consistent with most graphs in reality, then we sequentially store the incoming $v$ into the small slots.
2) When the node degree exceeds the specified number of small slots: these small slots merged in pairs to form several large slots with one of them deposited into the first pointer to the first (small) cuckoo hash table (S-CHT) that has just been activated, and all $v$ is transferred into that large-capacity S-CHT to accommodate more incoming $v$.
3) When the node degree is even larger: S-CHT is incremented with some regularity, 
to cope with the large increase of $v$;
and of course, it can also be decremented with some regularity to handle $v$ deletions.
In short, our \textsc{transformation} smoothly expands or shrinks the space through a series of spatial transformations to adapt to the increase or decrease of $v$, while ensuring that few accesses are required even in the worst case.
It greatly reduces the number of pointers and makes full use of limited memory space while ensuring speed.


Although we seem to fully guarantee time and space efficiency through the well-designed L/S-CHT, the shortcomings of cuckoo hashing itself have not yet been addressed.
As the memory space becomes tight due to the increase in incoming items, item replacement caused by hash collisions may occur frequently and may cause insertion failures.
On the one hand, not handling insertion failures may result in \alg{} no longer being error-free, which is unacceptable; on the other hand, we can also address it by directly expanding \alg{} every time an insertion failure occurs, but this may make it slower.
Therefore, we further propose the \textsc{denylist} optimization, which aims to cooperate with the \textsc{transformation} technique to efficiently accommodate those items that fail to be inserted.
Our ablation experiments have verified that this optimization can improve insertion and query throughput with almost no additional memory overhead.
For more details on \alg{}'s \textsc{transformation} technique and \textsc{denylist} optimization, please refer to $\S$~\ref{b-structure} and $\S$~\ref{BS}, respectively.
Further, we also propose an extended version of \alg{} for streaming scenarios to support duplicate edges in $\S$~\ref{ev}.

The rest of this paper is organized as follows.
We theoretically prove that the time and memory cost of \alg{} is desirable through mathematical analysis in $\S$~\ref{MSA}.
We conduct extensive experiments on 7 large-scale graph datasets with different characteristics to evaluate the performance of \alg{} on basic tasks and graph analytics tasks.
The results clearly show that \alg{} has the fastest speed and the lowest memory overhead on almost all tasks.
Finally, \alg{} is integrated into Redis and Neo4j databases as use cases.
See $\S$~\ref{sec:res} for more details.
All relevant source code is already open source on Github \cite{link}.

\bbb{Main Experimental Results:}
1) For basic tasks ($\S$~\ref{exp:TM}), the insertion and query throughput of \alg{} is on average $32.66\times$ and $133.62\times$ faster than those of Spruce, respectively, while its memory usage is  on average $1.47\times$ less than that of Spruce ($i.e.$, only 68.03\% of its);
2) For graph analytics tasks ($\S$~\ref{exp:Analytics}), the running time of \alg{} on 7 typical tasks (Breadth-First Search, Single-Source Shortest Paths, Triangle Counting, Connected Components, PageRank, Betweenness Centrality, and Local Clustering Coefficient) is on average $0.73\times$, $168.45\times$, $21.33\times$, $1.07\times$, $1.03\times$, $16.17\times$, and $5.80\times$ faster than those of Spruce, respectively.



%% file: VLDB/2-rw.tex
\section{Related Work}
\label{sec:related}

In this section, we introduce graph storage schemes whose main contribution lies in data structure design and Cuckoo Hash Table (CHT)—the most basic data structure of \alg{}, in $\S$~\ref{RW:exist} and $\S$~\ref{RW:CH}, respectively.

\subsection{Existing Solutions}
\label{RW:exist}

There are many existing works with the concept of dynamic graph storage, only some of which focus on optimizing graph updates (insertion, deletion, and attribute change, $etc.$) at the algorithm level.
Most of them have data structures based on or improved on adjacency lists or CSRs.
Of course, there are also other or hybrid ones.
We only select representatives to introduce for the sake of space.



\bbb{Adjacency list-based: }GraphOne \cite{kumar2020graphone} uses a complementary combination of adjacency lists and edge log lists. The adjacency list stores snapshots of old data, while the edge log records the latest updates, which periodically transfers data into the adjacency list in batches.
LiveGraph \cite{zhu2020livegraph} uses the proposed Transactional Edge Log (TEL) and Vertex Blocks (VB) to store edge information and nodes, respectively. In TEL, edge insertions and updates are performed in the form of log entries in a specified order.
RisGraph \cite{feng2021risgraph} uses the proposed indexed adjacency lists and sparse arrays. The indexed adjacency list is a dynamic array including arrays and edge indexes for continuous storage, while the sparse array is used for updates to avoid unnecessary data access.
Sortledton \cite{fuchs2022sortledton} uses a customized adjacency list, including the expandable adjacency index and adjacency set, to store the mapping from nodes to edge sets and the neighbors of each node, respectively.
Wind-Bell Index (WBI) \cite{qiuwind} consists of an adjacency matrix and many adjacency lists, where each bucket of the matrix is associated with an adjacency list through a pointer.
It selects the shortest hanging list through multiple hashes to address the slow query caused by node degree imbalance that is not considered in existing graph databases.
Spruce \cite{shi2024spruce} consists of an edge-storage part and a node-indexing part similar to the vEB tree.
The edge-storage part is based on the adjacency list and is used to store the edges as well as attributes.
The node-indexing part includes a hash table and a bit vector. It is used to record node identifiers and map nodes to their connected edges.
Also, it divides the 8-byte identifier into 4, 2, 2, where 4 is stored in the hash table to share the same hash address, and two 2s are stored in the bit vector associated with the edge storage part.
In this way, Spruce achieves low memory consumption and efficient dynamic operations, but it still needs to record quite a few pointers to any graph.

\begin{table}[t]
\begin{center}
\renewcommand\arraystretch{1.2}
\caption{Symbols used in this paper.}
\vspace{0.03in}
\label{taDLe:symbol}
\resizebox{.99\columnwidth}{!}{
\begin{tabular}{|m{0.2\columnwidth}<{\centering}|m{0.75\columnwidth}|}
\hline
\textbf{Notation}&\textbf{Meaning}\\
\hline
$\left \langle u, v \right \rangle$& A distinct graph item\\
\hline
L-CHT& The large cuckoo hash table\\
\hline
$H_1(.)$, $H_2(.)$& Two hash functions associated with L-CHT\\
\hline
S-CHT& The small cuckoo hash table\\
\hline
$h_1(.)$, $h_2(.)$& Two hash functions associated with S-CHT(s)\\
\hline
$n$& The length of 1st S-CHT\\
\hline
$R$& The number of large slots in Part 2 of each cell\\
\hline
$LR$& The loading rate\\
\hline
$G$& The preset $LR$ threshold for expansion\\
\hline
$\Lambda$& The preset overall $LR$ threshold for contraction\\
\hline
L-DL& Denylist for L-CHT(s)\\
\hline
S-DL& Denylist for S-CHT(s)\\
\hline
$T$ & Maximum number of loops in L/S-CHT \\
\hline
$w$ & Weight, or number of times $\left \langle u, v \right \rangle$ is repeated \\
\hline
\end{tabular}}
\end{center}
\vspace{-0.2in}
\end{table}

\bbb{CSR-based: }To address the problem that CSR is not suitable for dynamic graphs, PCSR \cite{wheatman2018packed} replaces the array storing neighbors in CSR with the packed memory array (PMA)\footnote{PMA is a dynamic array used to maintain an ordered collection of items. It balances item insertion and deletion operations by interspersing empty slots within the array to optimize access and modification performance.} \cite{bender2007adaptive}, which essentially maintains an implicit complete binary search tree.
To avoid frequent rebalance of PCSR, VCSR \cite{islam2022vcsr} maintains PMA by reserving empty slots between nodes in proportion to the current node degree.
Teseo \cite{de2021teseo} mainly uses PMAs as leaf nodes of improved B+ trees for graph updates, but only supports undirected graphs. Each PMA is divided into several expandable segments, and a hash table maps nodes to locations within the segment.

\bbb{Other-based: }Terrace \cite{pandey2021terrace} decides which data structure to use for storage based on the node degree: 1) Nodes with few neighbors are stored in an sorted array; 2) Nodes with a medium number of neighbors are stored in a PMA; 3) Nodes with many neighbors are stored in a B+ tree.
However, its biggest drawback is that the number of nodes for the workload must be known in advance.

\subsection{Orthogonal Work}
\label{RW:Orth}

Recently, there are notable works such as VEND \cite{li2023vend} that aim to accelerate edge queries by filtering no-result edges, which are orthogonal to graph storage/databases.
Since most node pairs in the real world have only no-edge connections, VEND introduces a novel data structure for nodes to store redundant neighbor information based on range and hash solutions.
We leave the possibility of applying VEND to \alg{} as future work.

\subsection{Cuckoo Hashing}
\label{RW:CH}

Cuckoo Hashing (or Cuckoo Hash Table, CHT) \cite{pagh2004cuckoo} contains two hash tables, each associated with a hash function.
Each newly inserted item is mapped to two candidate buckets (one in each table), and one of the two buckets will be selected to store it.
If at least one of the two candidate buckets is empty, CHT stores the item in an empty bucket.
If both candidate buckets are full, CHT randomly selects one of the stored items to kick out, and the kicked out item will be re-inserted into its another candidate bucket.
This process is repeated until each item finds a bucket to settle down, or reaches the maximum kick-out threshold and has to exit, which means that there is an item insertion failure.
To query an item, CHT only needs to directly check the two candidate buckets through two hash functions.
Therefore, CHT has a high loading rate and $O(1)$ query time complexity.
However, in the worst case, item insertions take a lot of time while still failing as described above, especially when CHT has a high loading rate.

%% file: VLDB/3-alg.tex
\section{\alg{} Design}
\label{sec:alg}

In this section, we present the basic version of CuckooGraph that does not support duplicate edges in $\S$~\ref{bv} and the extended version of CuckooGraph that supports duplicate edges in $\S$~\ref{ev}, respectively.
The symbols (including abbreviations) frequently used in this paper are shown in Table \ref{taDLe:symbol}.

\begin{figure*}
    \centering    \includegraphics[width=0.99\linewidth]{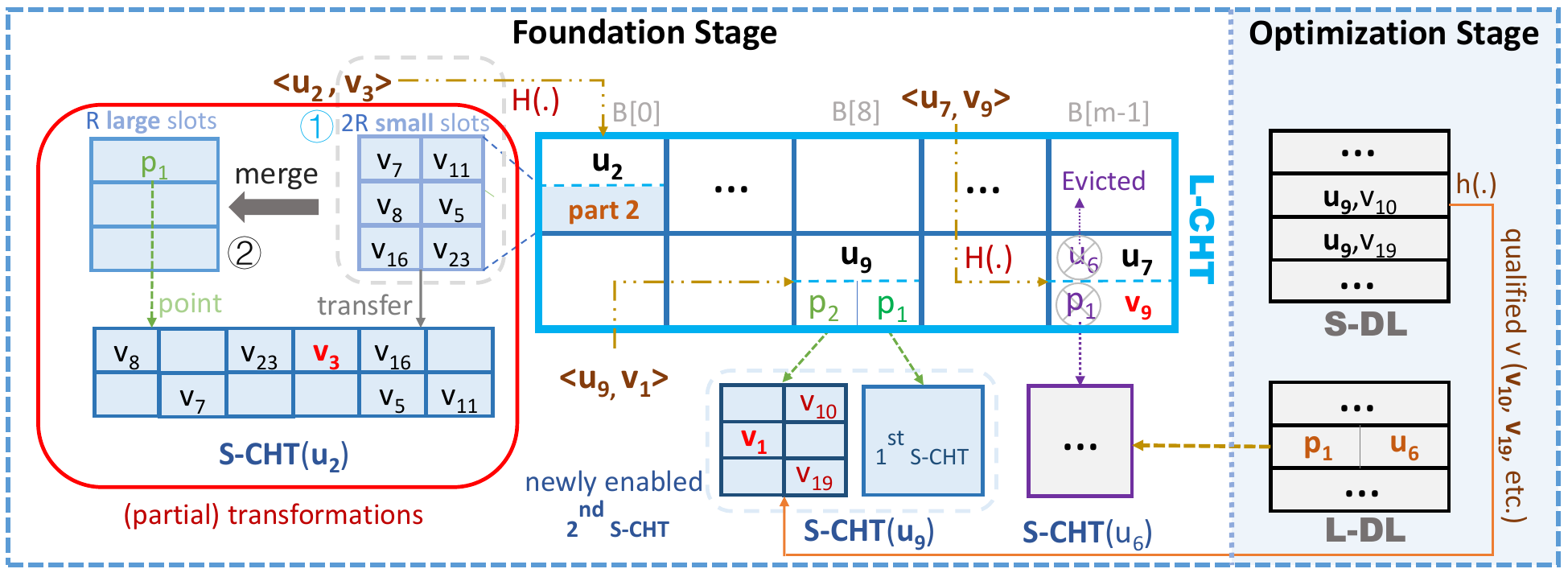}
    \vspace{0.05in}
    \caption{Data structure and examples of CuckooGraph. For clarity, only one hash table per L/S-CHT is shown.}    
    \label{structure}
\end{figure*}

\subsection{Basic Version}
\label{bv}

\subsubsection{\textbf{Foundation Stage (Transformable Data Structures)}}~
\label{b-structure}

As shown in Figure \ref{structure} (Foundation Stage), the basic data structure of \alg{} consists of one (or more) large cuckoo hash table(s) (denoted as L-CHT(s)) and many small cuckoo hash tables (denoted as S-CHTs) associated with it/them, all of which are specially designed.

L-CHT has two bucket arrays denoted as $B_1$ and $B_2$, respectively, associated with two independent hash functions $H_1(.)$ and $H_2(.)$, respectively\footnote{The basic structure of S-CHT is no different from that of L-CHT, except for the length and the association with two other independent hash functions $h_1(.)$ and $h_2(.)$.}.
Each bucket has $d$ cells, each of which is designed to have two parts: Part 1 and Part 2.
For any arriving graph item $\left \langle u, v \right \rangle$, assuming it is mapped to bucket $B_H$ and deposited in the cell $C$, then:
Part 1 of $C$ is used to store $u$, while Part 2 of $C$ is directly used to store $v$ \textbf{or} the pointer(s) pointing to S-CHT which is used to actually store $v$.
In Foundation Stage (red box), Part 2 is designed as a structure that can be flexibly transformed in a manner determined by \textit{the number of $v$} (denoted as $l$) corresponding to $u$ in Part 1, as shown below: 
\ding{172} Part 2 is initialized to $2R$ small slots, $i.e.$, up to $2R$ $v$ can be recorded, to handle the situation where $l\leq2R$;
\ding{173} When $l>2R$, $2R$ small slots are merged in pairs to form $R$ large slots dedicated to storing $R$ pointers usually with more bytes, and 1st large slot is deposited with a pointer that points to 1st S-CHT; 
Then, all the current $v$ are stored into this S-CHT of length\footnote{We define the length of CHT as the number of buckets in the array with more buckets.} $n$.


Next, we proceed to describe the transformation strategies of S-CHT and L-CHT to efficiently cope with the increasing $l$ as follows:
1) If the growing $l$ causes the loading rate ($LR$) of the 1st S-CHT to reach the preset threshold $G$ before the current $v$ arrives, we enable the 2nd pointer and store it in the 2nd large slot, as well as simultaneously enables the 2nd S-CHT;
2) By analogy, when $LR$ of the $(R-1)$-th S-CHT also reaches $G$, we continue to enable the $R$-th pointer and S-CHT similarly.
Here, we allocate the length of each newly enabled S-CHT in a memory-efficient manner, which is specifically related to the $R$ value.
We illustrate the transformation rule of length with $R = 3$, when there are at most 3 S-CHTs, as shown in Table \ref{Transformation-Rule}.
When $0\rightarrow1$ and $1\rightarrow2$ occur, the 2nd and 3rd S-CHTs with a length of $0.5n$ are enabled in turn; 
When $2\rightarrow3$ occurs, the 1st, 2nd and 3rd S-CHTs are merged at once into a new 1st S-CHT of length $2n$ on the 1st pointer, and the new 2nd S-CHT with length $n$ is enabled on the 2nd pointer; and so on.
In summary, different $R$ values correspond to different transformation rules, and such rules can also be applied to L-CHT to better handle unpredictable large-scale graphs.

\bbb{Reverse Transformation: }We introduce reverse transformation strategies for S-CHT and L-CHT to efficiently cope with the decreasing $l$.
Similar to the situation where $l$ increases, if the deletion of the current $v$ happens to cause the overall loading rate of the S-CHT chain\footnote{For convenience, we call all S-CHT(s) associated with the pointers corresponding to each $u$ a S-CHT chain.} to be lower than another threshold $\Lambda$, then we delete/compress the S-CHT where the $v$ was originally located as follows:
1) If there are two or more S-CHTs on the S-CHT chain, we delete the current S-CHT and transfer the previously stored $v$ on it to other S-CHTs;
2) If there is only one S-CHT left on the S-CHT chain, we compress the length of the S-CHT to half of the original length.
\textit{The above processing can be applied similarly to L-CHT.}

\begin{table}[t]
\renewcommand\arraystretch{1.1}
\centering
\caption{An Example of Transformation Rule}
\vspace{0.02in}
\label{Transformation-Rule}
\begin{tabular}{cccc}
\hline \# $LR>G$ & the 1st S-CHT & the 2nd S-CHT & the 3rd S-CHT \\
\hline 
0&\texttt{n} & \texttt{null} & \texttt{null} \\
1&\texttt{n} & \texttt{n/2} & \texttt{null} \\
2&\texttt{n} & \texttt{n/2} & \texttt{n/2} \\
3&\texttt{2n} & \texttt{n} & \texttt{null} \\
4&\texttt{2n} & \texttt{n} & \texttt{n} \\
5&\texttt{4n} & \texttt{2n} & \texttt{null} \\
6&\texttt{4n} & \texttt{2n} & \texttt{2n} \\
7&\texttt{8n} & \texttt{4n} & \texttt{null} \\
\texttt{\ldots} & \texttt{\ldots} & \texttt{\ldots} & \texttt{\ldots} \\
\hline
\end{tabular}
\vspace{-0.1in}
\end{table}

\subsubsection{\textbf{Optimization Stage (Denylist)}}~
\label{BS}

There is one aspect of our design that has not been considered so far: the original CHT may suffer from item insertion failures.
A straightforward solution is to extend CuckooGraph via the transformations described above whenever an insertion failure occurs.
To address this issue more efficiently, we further propose an optimization called \textsc{Denylist} (DL), as shown in Figure \ref{structure} (Optimization Stage).

DL is actually a vector with a size limit.
In CuckooGraph, all S-CHT(s) and L-CHT(s) are each equipped with a DL, denoted as S-DL and L-DL, respectively.
However, S-DL and L-DL are organized differently:
1) Each unit of S-DL records a complete graph item, $i.e.$, a $\left \langle u, v \right \rangle$ pair;
2) While L-DL's is consistent with that of each cell of L-CHT(s), so that even if $u$ is kicked out during item replacement, the associated S-CHT(s) does not need to be copied/moved.
S-DL and L-DL are used to accommodate those that are ultimately unsuccessfully inserted into S-CHT(s) and L-CHT(s), respectively.
Let’s take S-DL, which cooperates with S-CHT(s), as an example for a more detailed explanation, as follows: 
1) Initially, we assume that $v$ is attempted to be inserted into an arbitrary S-CHT;
2) When the total number of kicked out exceeds the threshold $T$ and there is still an unsettled $v'$, the insertion fails, then $v'$ and its corresponding $u'$ is placed in S-DL;
3) Each time it is the S-CHT's turn to expand, we insert those $v''$ in S-DL whose $u''$ exactly match the $u''$ present in the current S-CHT into the new S-CHT.
4) Subsequently, S-DL continues to accommodate all failed insertion items as usual.

\subsubsection{\textbf{Operations}}~
\label{b-operations}

By introducing the operations of \alg{} below, we aim to show how it can handle dynamically updated large-scale graphs.

\bbb{Insertion: }The process of inserting a new graph item $e=\left \langle u, v \right \rangle$ mainly consists of three steps, as follows:

\begin{itemize}[leftmargin=1em]

\item Step 1: We first query whether $e$ is already stored in \alg{} through the \textbf{Query} operation below. If so, $e$ is no longer inserted; otherwise, proceed to Step 2.

\item Step 2: By calculating hash functions, we map $u$ to a bucket $B_H$ of L-CHT and try to store it in one of the cells. There are three cases here: 
\ding{172} $u$ is mapped to $B_H$ for the first time and there is at least one empty cell in the bucket. 
Then, we store $u$ in Part 1 of an arbitrary empty cell, and store $v$ in Part 2. See $\S$~\ref{b-structure} for extensions of related data structures that may be triggered by the arrival of $v$.
\ding{173} $u$ is mapped to $B_H$ for the first time but the bucket is full. 
Then, we randomly kick out the resident $u'$ in one of the cells and store $u$ in it. The remaining operations are the same as \ding{172}. For this kicked-out $u'$, we just re-insert it.
\ding{174} If $u$ has been recorded in $B_H$, we directly store $v$ in Part 2.

\item Step 3: For any $u$ and $v$ that were not successfully inserted into Part 1 and Part 2 (in case of S-CHT), respectively, we store the relevant information in L-DL and S-DL, respectively.

\end{itemize}

%
Next, we illustrate the above insertion operations through the examples in Figure \ref{structure}.
For convenience, we assume that $R=3$ and that there is only one bucket array for each L-CHT and S-CHT, associated with hash functions $H(.)$ and $h(.)$, respectively.

\xxx{Example 1: }For the new item $\left \langle u_2, v_3 \right \rangle$, it is mapped to $B[0]$.
There is one cell in $B[0]$ that has already recorded $u_2$, but the small slots in Part 2 has recorded $2R$ $v$ in total.
Then, we transfer all $v$ to the 1st S-CHT by computing the hash function $h(.)$.

\xxx{Example 2: }For the new item $\left \langle u_9, v_1 \right \rangle$, it is mapped to $B[8]$.
$u_9$ has been recorded in $B[8]$, but $LR$ of the 1st S-CHT exceeds $G$. Then, we map the qualified $v$ ($v_{10}$ and $v_{19}$) in S-DL and $v_1$ to the 2nd S-CHT with half the previous length.

\xxx{Example 3: }For the new item $\left \langle u_7, v_9 \right \rangle$, it is mapped to $B[m-1]$.
Since $B[m-1]$ is already full, we randomly kick out an unlucky $u_6$ and stores $u_7$ in Part 1 of the new empty cell, and stores $v_9$ in the 1st small slot of Part 2.
Then, we try to map $u_6$ into another bucket. 
Suppose one $u$ is not settled in the end and it happens to be $u_6$. 
We have to place it in L-DL, along with the pointer associated with it.

\bbb{Query: }The process of querying a new graph item $e=\left \langle u, v \right \rangle$ mainly consists of two steps, as follows:

\begin{itemize}[leftmargin=1em]

\item Step 1: \textit{Query whether $u$ is in L-CHT, otherwise check whether it is in L-DL.} 
Specifically, we first calculate hash functions to locate a bucket $B_H$ that may record $u$, and then traverse $B_H$ to determine whether $u$ exists. If so, we directly execute Step 2; otherwise, we further query L-DL: if $u$ is in L-BS, proceed to Step 2; otherwise, return \texttt{null}.

\item Step 2: \textit{Query whether $u$ has a corresponding $v$ in Part 2.} 
If so, we directly report it; otherwise, we further query S-DL: if $v$ is in S-DL, report it; otherwise, return \texttt{null}.

\end{itemize}

\bbb{Deletion: }To delete a graph item, we first query it and then delete it.
For compression of related data structures that may be caused by deleting this item, see \textbf{Reverse Transformation} in $\S$~\ref{b-structure}.

\subsection{Extended Version}
\label{ev}

The base version of \alg{} can be easily extended into a new version that efficiently supports storing duplicate edges, as designed for streaming scenarios.

\bbb{Data Structure: }We only need to make a few customized modifications based on the transformable data structure proposed in Section \ref{b-structure}.
Specifically for S-CHT, each small slot in Part 2 needs to change from storing only $v$ to storing both $v$ and weight $w$. As more information is recorded, the number of small slots changes from $2R$ to $R$ accordingly, $i.e.$, the space of two small slots is used to store $\left \langle v, w \right \rangle$.

Next, we describe the operations related to the weighted version of \alg{}, focusing only on its differences from the basic version for better intuition.

\bbb{Insertion: } The main difference from the basic version is that when it is initially discovered that the item $\left \langle u, v \right \rangle$ already exists, it changes from doing nothing to incrementing the corresponding $w$ by 1 (or other defined value, the same below) and then returning.

\bbb{Query: }Report the item and return the value of $w$.

\bbb{Deletion: }We decrement the $w$ of the item by 1 and delete the item when the weight is reduced to 0.

%% file: VLDB/4-math.tex
\section{Mathematical Analysis}
\label{MSA}

In this section, we theoretically analyze the performance of \alg{} (basic version). 
Specifically, we show its time and memory complexity.

\subsection{Time Cost of \alg{}}

In this part, we assume that there is only insertion operation in \alg{}. 
First, we show a theorem with respect to multi-cell cuckoo hash tables.
Then, we analyze the insertion time complexity of \alg{} based on it. 
Finally, we analyze the time cost of the expansion process.

\begin{table*}[t]
  \centering
  \renewcommand\arraystretch{1.3}
  \caption{Comparison of complexity between different solutions.}
  \resizebox{1.5\columnwidth}{!}{
  \begin{tabular}{c|c|c|c}
    \hline \multirow{2}{*}{ \textbf{Algorithm} } & \multicolumn{2}{c|}{ \textbf{Amortized Time Complexity} } & \multirow{2}{*}{ \textbf{Space Complexity} } \\
\cline { 2 - 3 } & \textbf{Insert Edge $\langle u, v\rangle$} & \textbf{Query Edge $\langle u, v\rangle$} \\
    \hline
    LiveGraph \cite{zhu2020livegraph} & $O(1)$  &  $O(\deg(v))$  & $O(|E|)$ \\
    Spruce \cite{shi2024spruce} & $O\left(\frac{|E|}{|V|}\right) $ & $O\left(\log \frac{|E|}{|V|}\right) $ &  $O(|E|)$   \\
    Sortledton \cite{fuchs2022sortledton} & $O(\log |E|)$ &  $O(\log|E|)$  & $O(|E|)$ \\
    WBI \cite{qiuwind} & $O(1)$ & $O(\frac{|E|}{K^2})$ & $O(K^2+|E|)$  \\
    \textbf{\alg{} (Ours)} & $O(1)$ &  $O(1)$  & $O(|E|)$\\
    \hline
  \end{tabular}}
  \label{tab:compare}
  \label{-0.1in}
\end{table*}

\begin{theorem}
\label{theorem1}
    Assume that a cuckoo hash table has $m$ buckets, each bucket has $d$ cells, and there are $n$ distinct items to insert. Let $dm=(1+\varepsilon)n$. If $d\geq \max \left\{ 8, 15.8\ln \frac{1}{\varepsilon} \right\}$, then the expected time complexity for inserting an item is $(\frac{1}{\varepsilon})^{O(\log d)}$.     
\end{theorem}

This theorem is based on the relevant proof in \cite{dietzfelbinger2007balanced}.

Since the $LR$ is defined as $\frac{n}{dm}$, by setting $\varepsilon = \frac{dm}{n}-1 \geq \frac{1}{G}-1$, the expected time cost for inserting an item can be calculated as $(\frac{G}{1-G})^{O(\log d)}$. 
If we set $T$ as the maximum number of loops for L-CHT, then the worst-case insertion time cost can be further written as $O(T)$, or $O(1)$ if $T$ is not very large. 
Here, we provide an experiment to verify the above: We expand \alg{} starting from the minimum length and insert all the edges of NotreDame Dataset into it in sequence. 
It can be calculated that the average number of insertions per item in L-CHT and S-CHT considering the expansion is about 1.017 and 1.006, respectively, which is much less than $T$ ($T=250$ in the experiments in $\S$~\ref{sec:res}).

Then, we analyze the amortized cost of inserting $N$ edges into \alg{}. We assume that two hash functions $H_1, H_2$ in L-CHT are the same modular hash functions. The insertion time complexity of \alg{} can be summarized as follows:

\begin{theorem}\label{theorem2}
    Assume that the L/S-DL are never full during insertion procedure and inserting an edge into L-CHT (not triggering L-CHT expansion) costs 1 dollar, then the price of inserting $N$ edges into L-CHT will not exceed $3N$ dollars, and its expectation will not exceed $2.25N$ dollars.  
\end{theorem}

\begin{proof}
    We first analyze the cost of merging and expansion: Assume that the 1st, 2nd, 3rd L-CHT stores $x, y, z$ distinct $u$ respectively. 
    When merging L-CHT, we re-hash every $u$ and re-insert it into the merged L-CHT if its hash value does not match its bucket index. The hash functions are the same modular hash functions, and the size of the merged L-CHT is 2 times larger than 1st L-CHT, and 4 times larger than 2nd and 3rd L-CHT. Hence, the probability to re-insert every $u$ is $\frac{1}{2}$ in 1st L-CHT, and $\frac{3}{4}$ in 2nd and 3rd L-CHT. In conclusion, the price of merging operation will not exceed $x+y+z$ dollars, and its expectation is $\frac{1}{2}x+\frac{3}{4}(y+z)$ dollars. 
    

    Then, we assume that after inserting $N$ edges, the 1st L-CHT has $2^kn$ cells, and it has $n$ cells before insertion. Assume that the 1st, 2nd, 3rd L-CHT stores $x_i, y_i, z_i$ distinct $u$ before $i$-th merging and expansion, then $x_i\leq 2^{i-1}Gn, y_i, z_i\leq 2^{i-2}Gn$. Hence, the total cost of merging and expansion will not exceed 
    \[
    2Gn+4Gn+8Gn+\cdots +2^{k}Gn = 2(2^k-1)Gn, 
    \]
    and its expectation is 
    \[
    \frac{5}{4}Gn+\frac{5}{2}Gn+\cdots +5\cdot 2^{k-3}Gn=\frac{5}{4}(2^k-1)Gn.
    \]
    If $N\leq 2Gn$, then the insertion costs $N$ dollars in total; otherwise, the $LR$ of L-CHT is smaller than $G$ but greater than $\frac{2}{3}G$, hence $N\geq \frac{2}{3}G\cdot (2^kn+2^{k-1}n) = 2^kGn$ and the total cost of merging and expansion will not exceed $2N$. In conclusion, the price of inserting $N$ edges into L-CHT will not exceed $N+2N=3N$ dollars, and its expectation will not exceed $N+\frac{5}{4}\cdot N=2.25N$. 
\end{proof}

\subsection{Memory Cost of \alg{}}
\label{sec:mem-cost}

In this part, we take both insertion and deletion operations into consideration. We first define a stable state for L/S-CHT and analyze its property. Then, we show the memory cost of \alg{} under the stable state. We assume that $\Lambda\leq \frac{2}{3}G$ in this part.

\begin{definition}
We define a group of L/S-CHTs as \textbf{stable}, if its overall loading rate ($LR$) is at least $\Lambda$. 
\end{definition}

The property of stable state is that, once a group of L/S-CHTs is stable, then it will be stable with high probability.

\begin{lemma}
    Assume that the number of graph items inserted into the L/S-CHTs at time $t$ are $l$ and $s$, respectively.
    If a group of L/S-CHTs is on stable state at time $t$, then it will always stay on stable state if the number of items in this group of L/S-CHTs is at least $l$ and $s$. 
\end{lemma}

\begin{proof}
    Since a group of L/S-CHTs is on stable state, its $LR$ must be greater than $\Lambda$. Then, once its $LR$ is greater than $G$, then the hash table will expand $\frac{4}{3}$ or $\frac{3}{2}$ times to its original size. 
    And once its $LR$ is less than $\Lambda$, the hash table will contract if possible. 
    As a result, if the number of items in this group of L/S-CHTs is at least $l$ and $s$, then the size of the L/S-CHTs will not be smaller after time $t$.
    Therefore, its $LR$ is still greater than $\Lambda$ after expansion. 
\end{proof}

\begin{theorem}
Assume that the L/S-DL are never full during insertion procedure and the L/S-CHTs are all on stable state. The upper bound of cells is $\frac{|V|}{\Lambda}$ for L-CHT and $\frac{|E|}{\Lambda}$ for all S-CHT (not including the L/S-DL), where $|V|$ denotes the number of distinct nodes, and $|E|$ denotes the number of distinct edges.  
\end{theorem}

\begin{proof}
We first analyze the number of cells in L-CHT: Since there are $|V|$ distinct nodes in the graph, they occupy at most $|V|$ cells in L-CHT. The lower bound of $LR$ is $\Lambda$ on stable state,  so L-CHT has at most $\frac{|V|}{\Lambda}$ cells. 

Then, we analyze the number of cells in S-CHT: Assume that $V=\{u_1, \cdots, u_{|V|}\}$, and the number of edges starting from $u_i$ is $f_i$. Since all groups of S-CHTs are on stable state, then the S-CHT for $u_i$ has at most $\frac{f_i}{\Lambda}$ cells. Hence, all S-CHT occupy at most $\frac{f_1}{\Lambda}+ \cdots +\frac{f_{|V|}}{\Lambda} \leq\frac{|E|}{\Lambda}$ cells.
\end{proof}



\subsection{Discussions}

In this part, we provide Table \ref{tab:compare}  to summarize the time and space complexities of \alg{} and some state-of-the-art schemes ($i.e.$, competitors mentioned in $\S$~\ref{sec:res}). Here, $K$ refers to the length/width of the matrix, which is a parameter of WBI.
In summary, our \alg{} has an insertion and query time complexity of $O(1)$ when Theorem \ref{theorem1} holds and $T$ is not very large, 
while its space complexity is still $O(|E|)$. 
The previous analysis in $\S$~\ref{sec:mem-cost} also proves that its space overhead is very small.

\vspace{-0.05in}

%% file: VLDB/5-newexp.tex
\section{Evaluation}
\label{sec:res}

In this section, we evaluate the performance of \alg{} through extensive experiments, which are briefly described as follows:
1) We introduce the experimental setup in $\S$~\ref{exp:Setup};
2) We evaluate how key parameters affect \alg{} in $\S$~\ref{exp:Settings};
3) We verify the effect of \textsc{Denylist} optimization through ablation experiments in $\S$~\ref{exp:Ablation};
4) We evaluate the insertion, query, and deletion throughput as well as memory usage of \alg{} and its competitors in $\S$~\ref{exp:TM};
5) We evaluate the running time of \alg{} and its competitors on graph analytics tasks (BFS, SSSP, TC, CC, PR, BC, LCC) in $\S$~\ref{exp:Analytics};
6) We deploy \alg{} on Redis and Neo4j databases and evaluate the speed in $\S$~\ref{Redis-I} and $\S$~\ref{Neo4j-I}, respectively.

\begin{table*}[h]
\centering
\renewcommand\arraystretch{1.2}
\caption{A brief analysis of the graph datasets used.}
\label{table:datasets}
\resizebox{1.5\columnwidth}{!}{
\begin{tabular}{lcrrrrrr}
\hline Graph Dataset & Weighted? & \# Nodes & \# Edges & \# Edges (dedup) & Avg. Deg. & Max. Deg. & Edge Density \\
\hline CAIDA & \ding{52} & 0.51M & 27.12M & 0.85M & $1.66$ & $17950$
 & $3.26\times10^{-6}$ \\
NotreDame & \ding{56} & 0.33M & 1.50M & 1.50M & $4.60$ & $10721$
 & $1.41\times10^{-5}$ \\
StackOverflow & \ding{52} & 2.60M & 63.50M & 36.23M & $13.92$ & $60406$ & $5.35\times10^{-6}$  \\
WikiTalk & \ding{52} & 2.99M & 24.98M & 9.38M & $3.14$ & $146311$
 & $1.05\times10^{-6}$ \\
Weibo & \ding{56} & 58.66M & 261.32M & 261.32M & $4.46$ & $278491$
 & $7.60\times10^{-8}$ \\
DenseGraph & \ding{56} & 8K & 57.59M & 57.59M & $7199.16$ & $14537$
 & $0.90$ \\
SparseGraph & \ding{56} & 5M & 30M & 30M & $6$ & $6$ & $1.20\times10^{-6}$ \\
\hline
\end{tabular}}
\end{table*}

\subsection{Experimental Setup}
\label{exp:Setup}

\bbb{Platform: }We conduct all the experiments on a 18-core CPU server (36 threads, Intel(R) Core(TM) i9-10980XE CPU @ 3.00GHz) with 128GB DRAM memory.
It has 64KB L1 cache, 1MB L2 cache for each core, and 24.75MB L3 cache shared by all cores.

\bbb{Implementation: }We implement \alg{} and the other competitors with C++ and build them with g++ 7.5.0 and -O3 option.
The hash functions we use are 32-bit Bob Hash (obtained from the open-source
website \cite{bobhash}) with random initial seeds. 
For \alg{}, we set $R=3$, as well as the ratio of the number of buckets in the two arrays of L/S-CHT is 2:1, and whether the basic or extended version of \alg{} is used depends on whether the dataset has repeated edges.
%

\bbb{Competitors: }Since there are many related works on dynamic graph storage, we rigorously select some of the SOTA ones from recent years for experimental comparison: 
LiveGraph \cite{zhu2020livegraph}, Sortledton \cite{fuchs2022sortledton}, Wind-Bell Index (WBI) \cite{qiuwind}, and Spruce \cite{shi2024spruce}.

\bbb{Datasets: }We use various graph datasets to comprehensively evaluate the performance of \alg{} and its competitors, and the details are shown in Table \ref{table:datasets}. 
1) The CAIDA dataset \cite{caida} is streams of anonymized IP traces collected by CAIDA. Each flow is identified by a five-tuple: source and destination IP addresses, source and destination ports, protocol. The source and destination IP addresses in the traces are used as the start and end nodes of the graphs, respectively.
2) The NotreDame dataset \cite{ndwg} is a web graph collected from University of Notre Dame. Nodes represent web pages, and directed edges represent hyperlinks between them.
3) The StackOverflow dataset \cite{Stack} is a collection of interactions on the stack exchange website called Stack Overflow. Nodes represent users and edges represent user interactions.
4) The WikiTalk dataset \cite{Wiki} is a collection of user communications obtained from English Wikipedia, and the nodes and edges refer to the same as above.
5) The Weibo dataset \cite{weibo} is captured from Sina Weibo Open Platform APIs, and the definitions of nodes and edges are similar to those of StackOverflow.
6) We synthesize the DenseGraph dataset.
7) We synthesize the SparseGraph dataset.

\bbb{Metrics: }We use the following key metrics.

\begin{itemize}[leftmargin=1em]

\item \textbf{Throughput: }It is defined as Million Operations Per Second (Mops). We use Throughput to evaluate the average insertion, query, and deletion speed.

\item \textbf{Memory Usage: }It is defined as the memory used to store a specified amount of edges.

\item \textbf{Running Time: }It is defined as the time spent performing the specified graph analytics tasks.

\end{itemize}



\begin{figure*}[t!]
\centering
    \subfigure[Insertion Throughput]{
    \includegraphics[width=4.2cm]{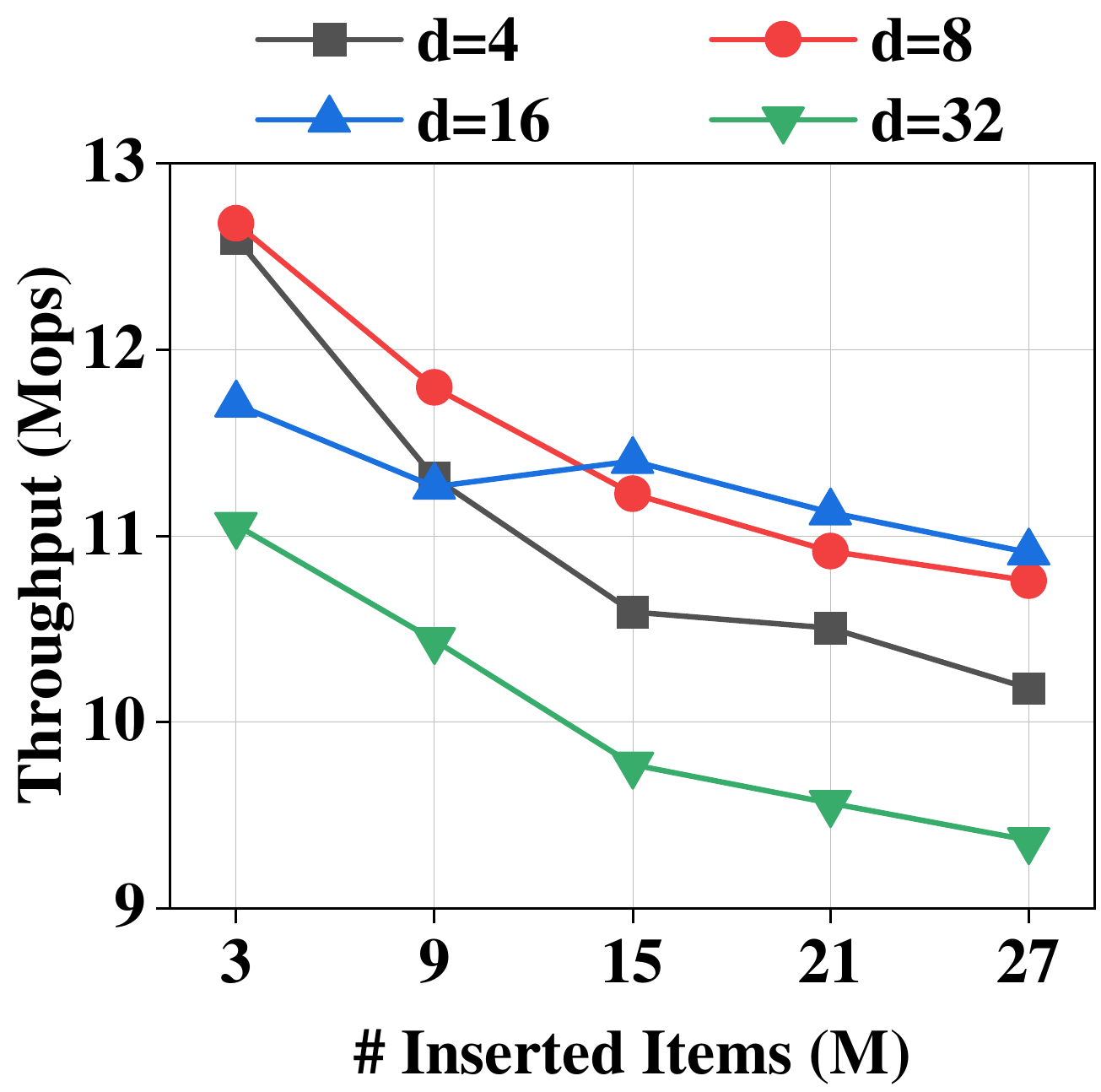}
    \label{para:d:it}
    }
    \subfigure[Query Throughput]{
    \includegraphics[width=4.2cm]{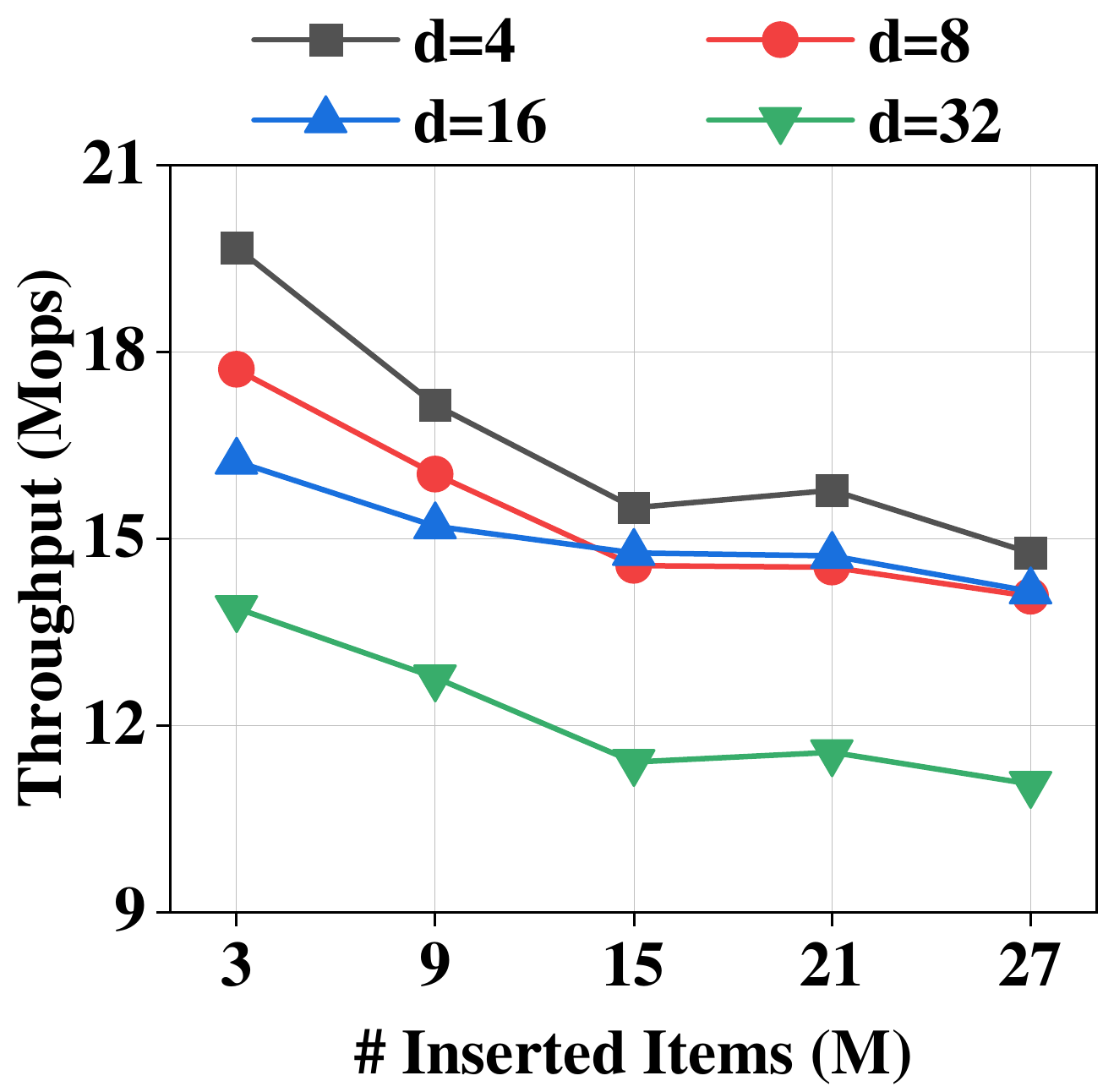}
    \label{para:d:qt}
    }
    \subfigure[Memory Usage]{
    \includegraphics[width=4.1cm]{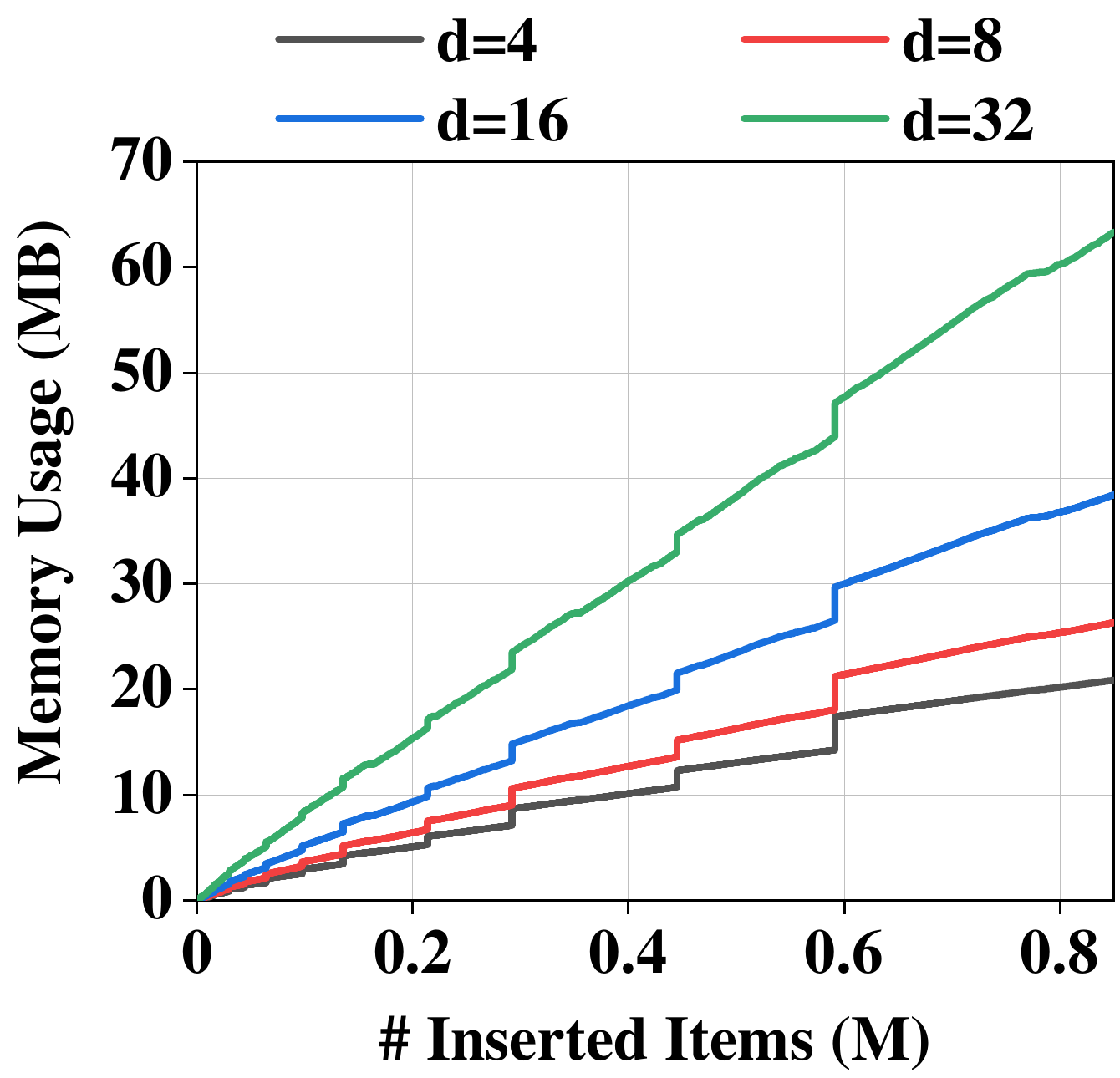}
    \label{para:d:mu}
    }
\caption{Tuning experiments for parameter $d$.}
\label{para:d}
\vspace{-0.1in}
\end{figure*}

\begin{figure*}[t!]
\centering
    \subfigure[Insertion Throughput]{
    \includegraphics[width=4.2cm]{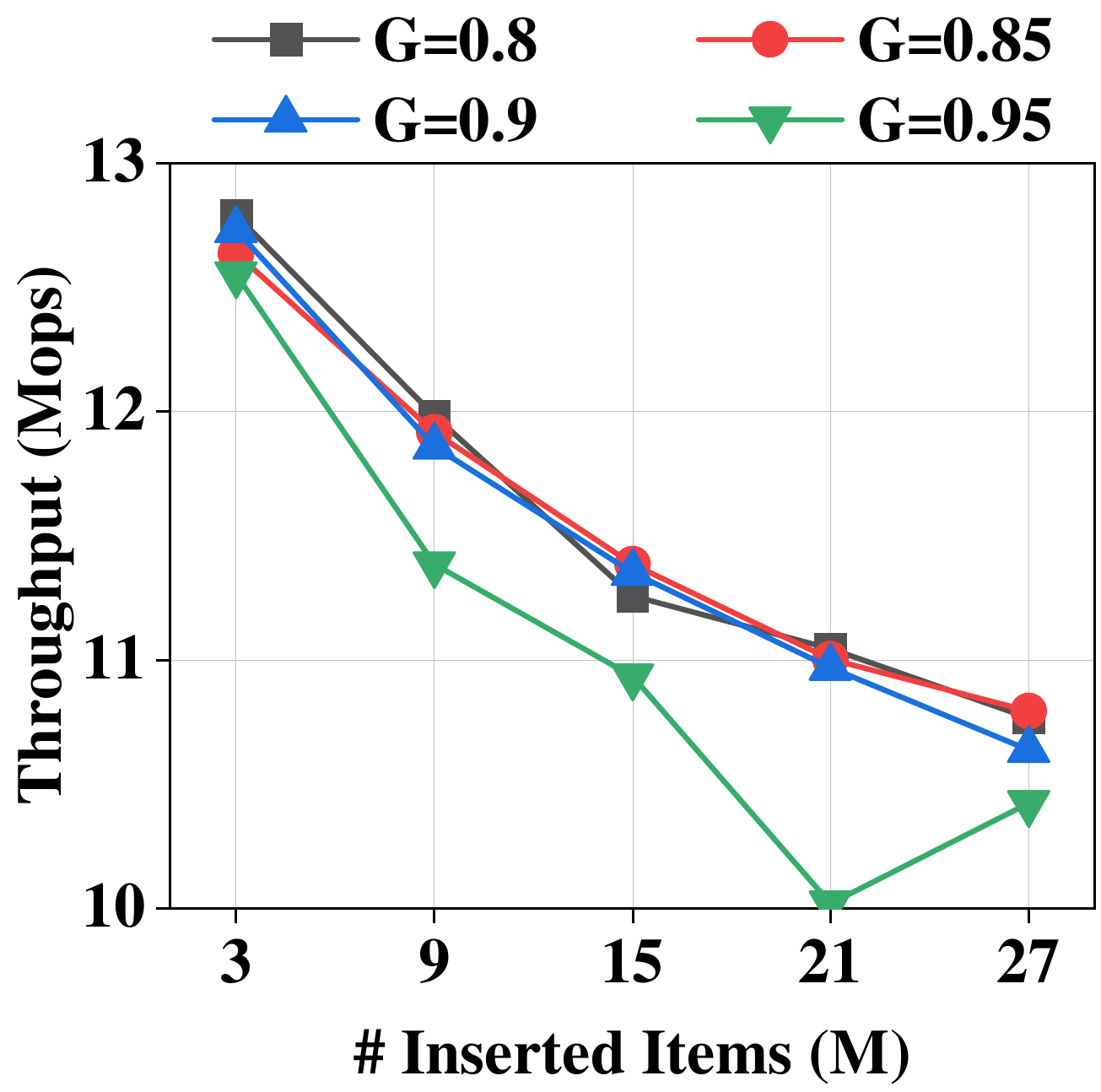}
    \label{para:G:it}
    }
    \subfigure[Query Throughput]{
    \includegraphics[width=4.2cm]{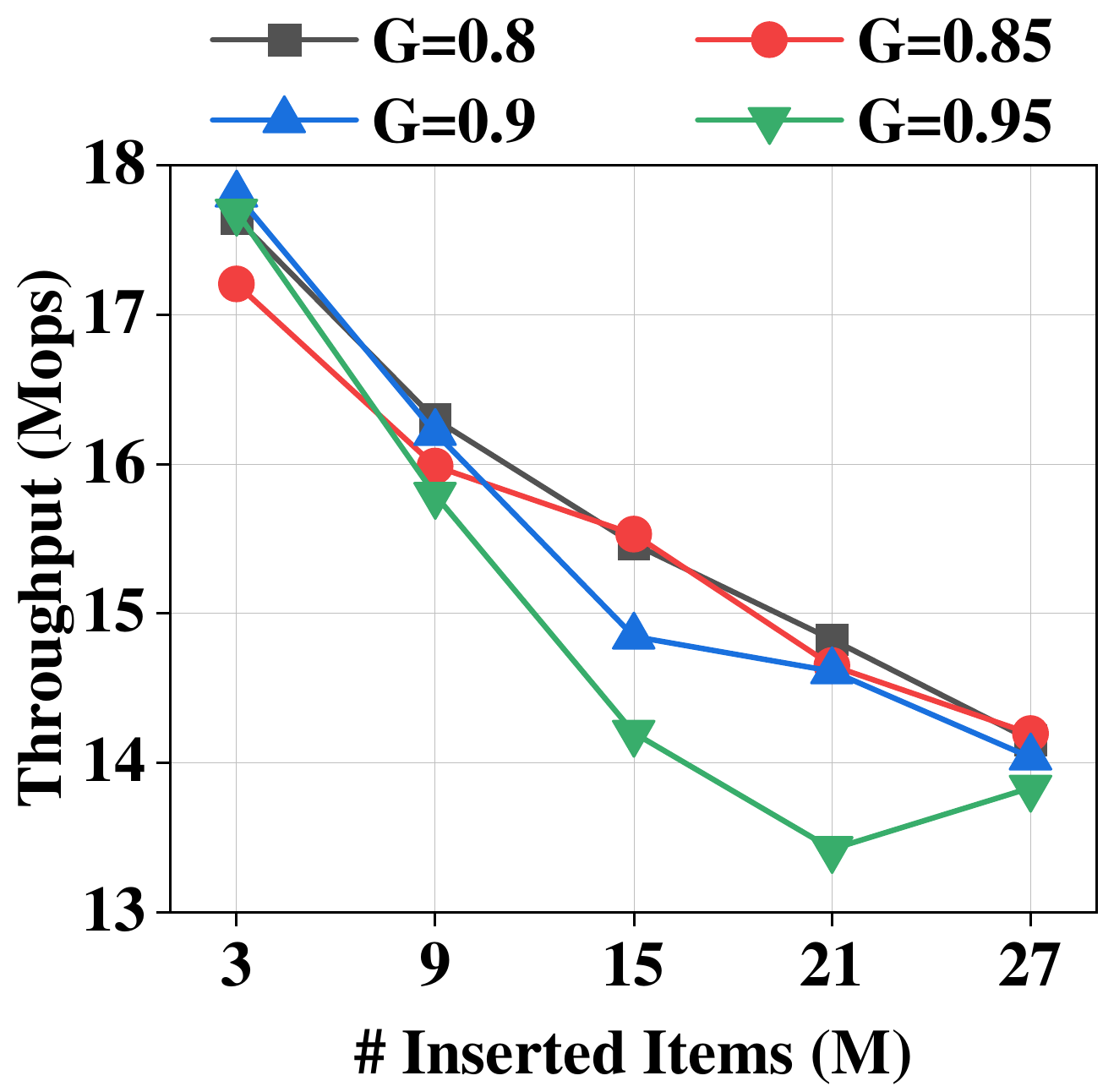}
    \label{para:G:qt}
    }
    \subfigure[Memory Usage]{
    \includegraphics[width=4.1cm]{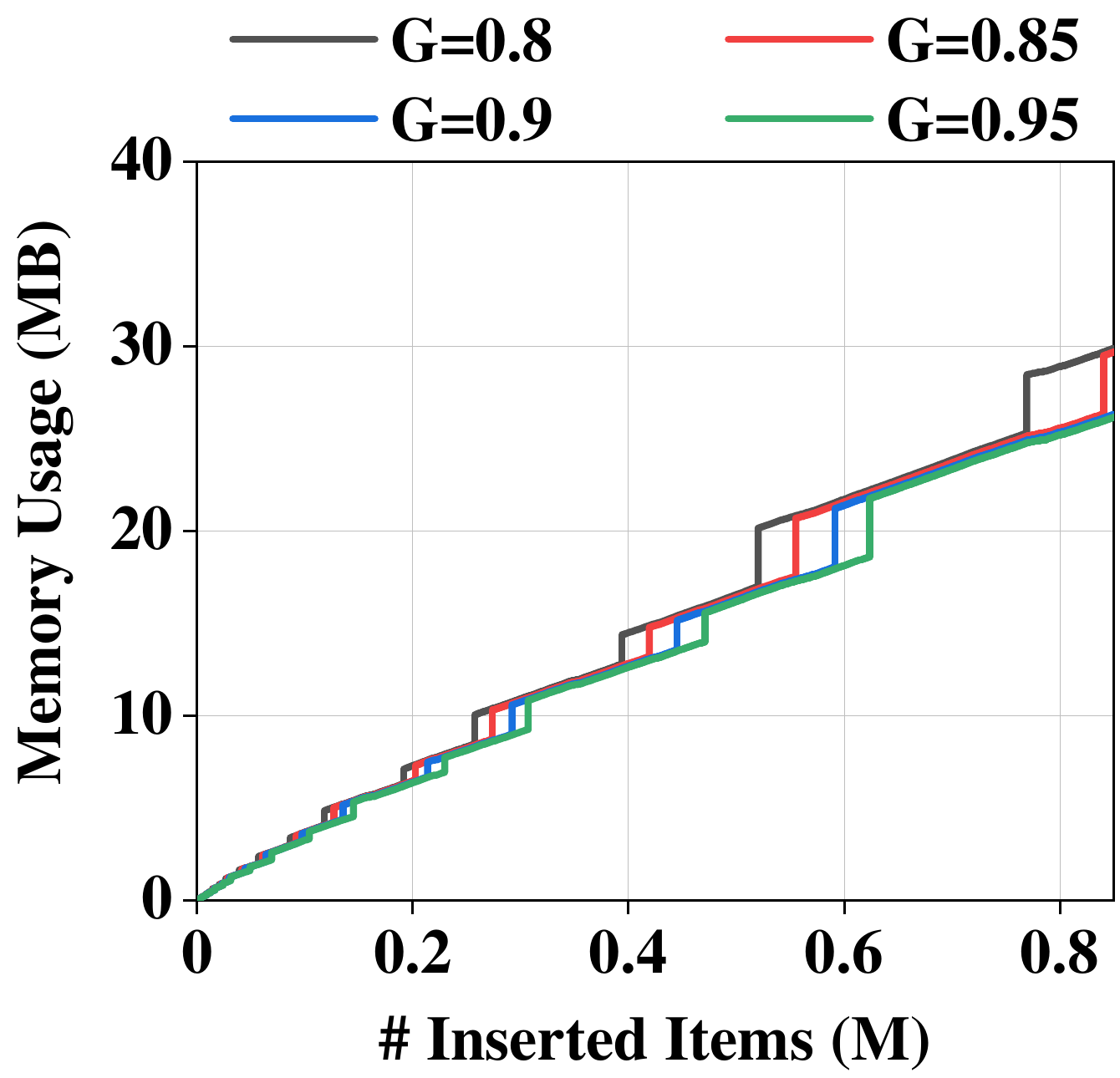}
    \label{para:G:mu}
    }
\caption{Tuning experiments for parameter $G$.}
\label{para:G}
\vspace{-0.1in}
\end{figure*}

\begin{figure*}[t!]
\centering
    \subfigure[Insertion Throughput]{
    \includegraphics[width=4.2cm]{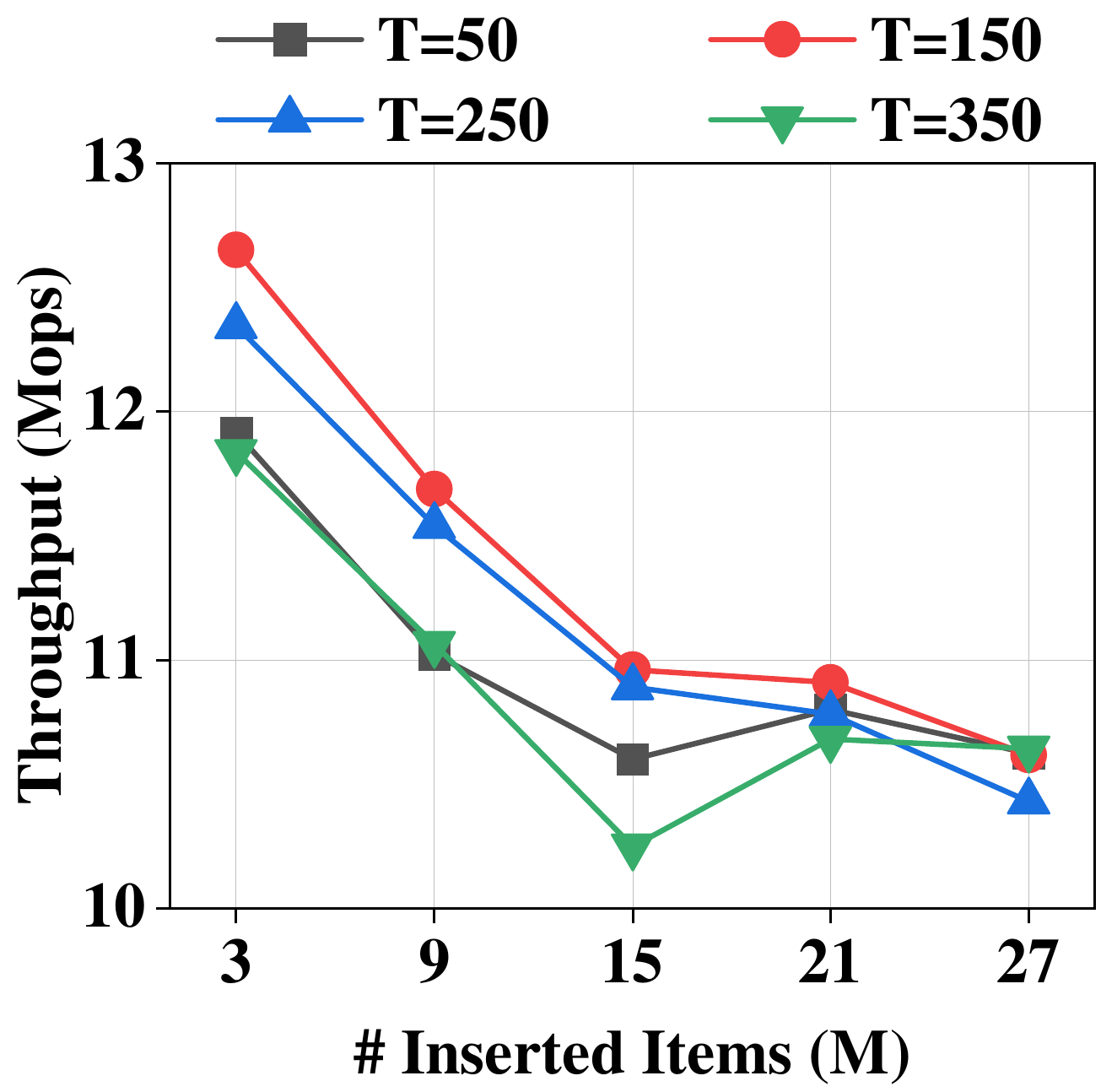}
    \label{para:T:it}
    }
    \subfigure[Query Throughput]{
    \includegraphics[width=4.2cm]{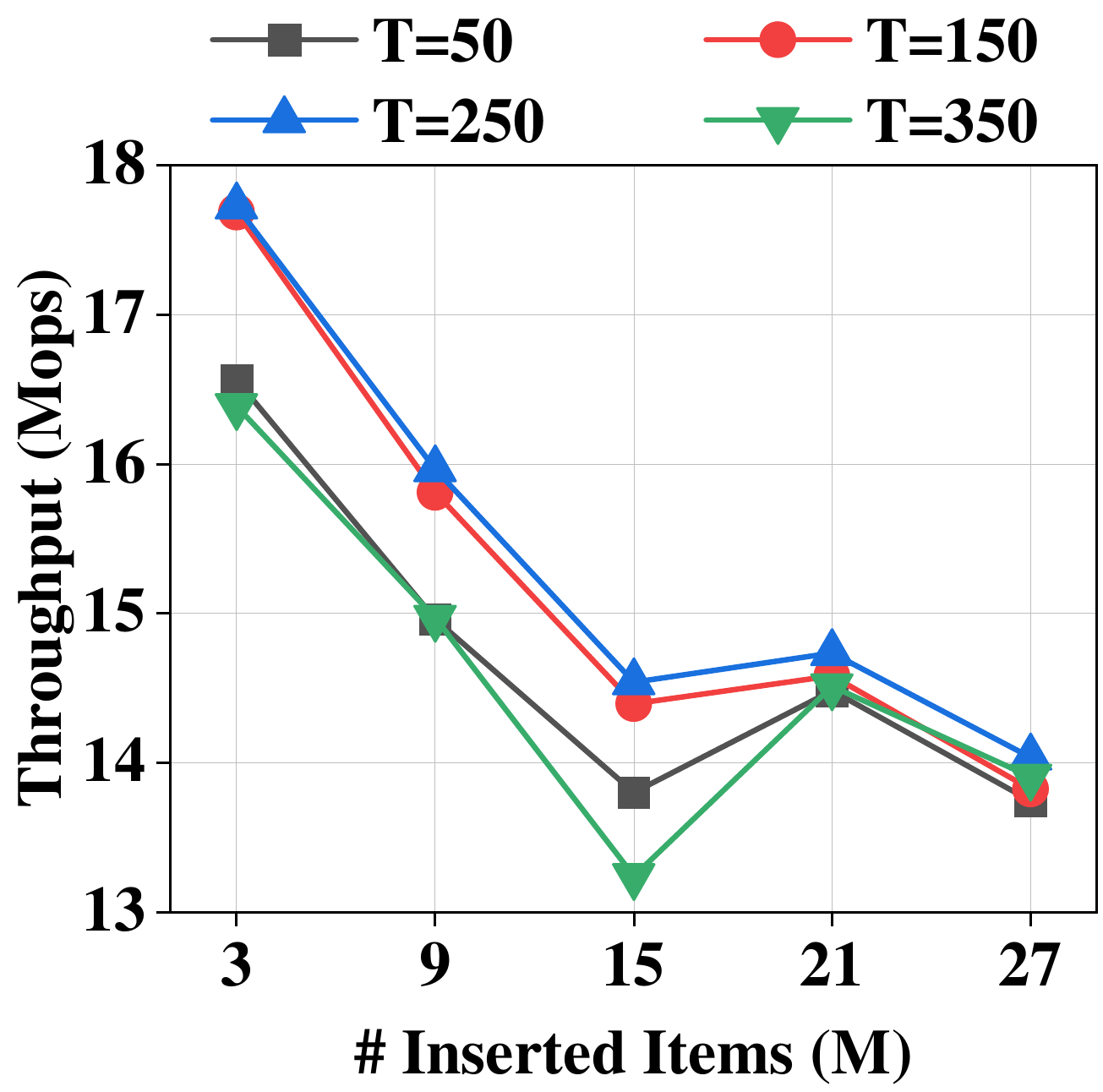}
    \label{para:T:qt}
    }
    \subfigure[Memory Usage]{
    \includegraphics[width=4.1cm]{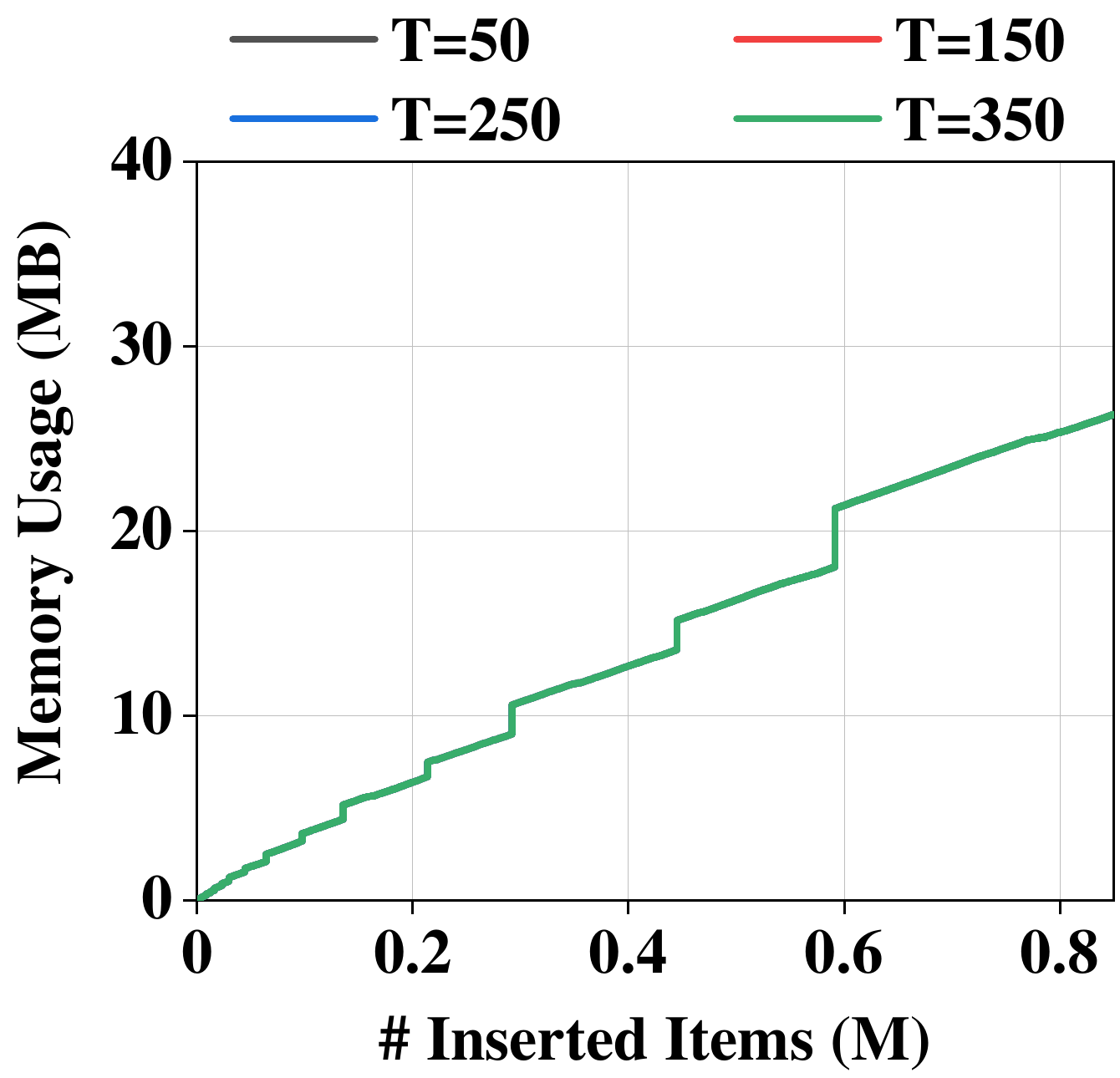}
    \label{para:T:mu}
    }
\caption{Tuning experiments for parameter $T$.}
\label{para:T}
\vspace{-0.1in}
\end{figure*}

\subsection{Experiments on Parameter Settings}
\label{exp:Settings}

In this subsection, we measure the effects of some key parameters for \alg{}, namely, the number of cells per bucket in L/S-CHT $\textbf{\textit{d}}$, the preset $LR$ threshold for expansion $\textbf{\textit{G}}$, and the maximum number of loops in L/S-CHT $\textbf{\textit{T}}$.
This experiment evaluates the effects by: 1) We first batch inserting edges in the CAIDA dataset into \alg{} and then batch querying them from \alg{}, and measure the average throughput separately; 2) We measure the memory usage by continuously inserting edges.

\bbb{Effects of $d$ (Figure \ref{para:d:it})-\ref{para:d:mu}: }
\textit{Our experimental results show that the optimal values of $d$ is 4 and 8.}
We find that $d=8$ and $d=4$ enable the fastest insertion and query throughput of \alg{}, respectively.
Also, the memory usage of \alg{} with $d=4$ and $d=8$ is the least and second least, respectively. 
Considering that smaller $d$ means smaller $LR$, we set $d=8$.

\bbb{Effects of $G$ (Figure \ref{para:G:it})-\ref{para:G:mu}: }
\textit{Our experimental results show that the overall performance is best when the value of $G$ is 0.9.}
We find that the insertion and query throughput of \alg{} with $G$ of 0.8, 0.85, and 0.9 are very close to each other, and all are faster than the one at $G$ of 0.95. 
In addition, the larger $G$ is, the smaller the memory usage of \alg{} is.
Thus, we set $G=0.9$ after the above considerations.

\bbb{Effects of $T$ (Figure \ref{para:T:it})-\ref{para:T:mu}: }\textit{The experimental results show that \alg{} achieves most ideal  performance at $T$ of 150 and 250.}
We find that \alg{} has the fastest insertion and query throughput when $T$ is 150 and 250, respectively.
Meanwhile, different values of $T$ make no difference to the memory usage of \alg{}.
Hence, we set $T=250$.

\subsection{Ablation Experiments}
\label{exp:Ablation}

In this subsection, we conduct ablation experiments to evaluate the individual effects of \textsc{Denylist} (DL) optimization on \alg{} performance to verify its effectiveness.
Our methodology is that every time an insertion failure occurs, we expand the size of \alg{} to $1.5\times$ its original size.
We use the CAIDA dataset and evaluate the effects in terms of insertion and query throughput as well as memory usage.
%


\begin{figure*}[t!]
\centering
    \subfigure[Insertion Throughput]{
    \includegraphics[width=4.4cm]{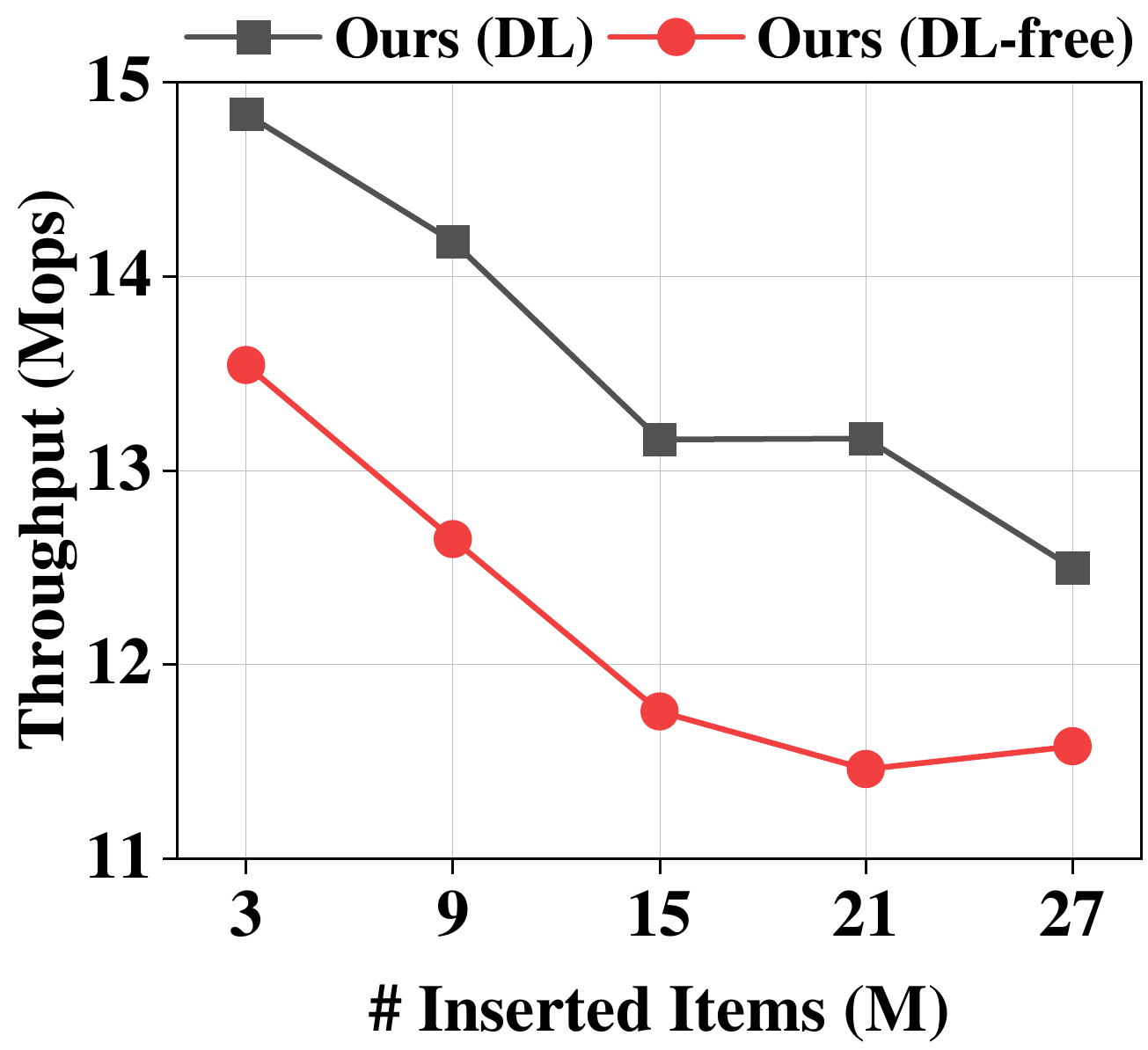}
    \label{DL:it}
    }
    \subfigure[Query Throughput]{
    \includegraphics[width=4.4cm]{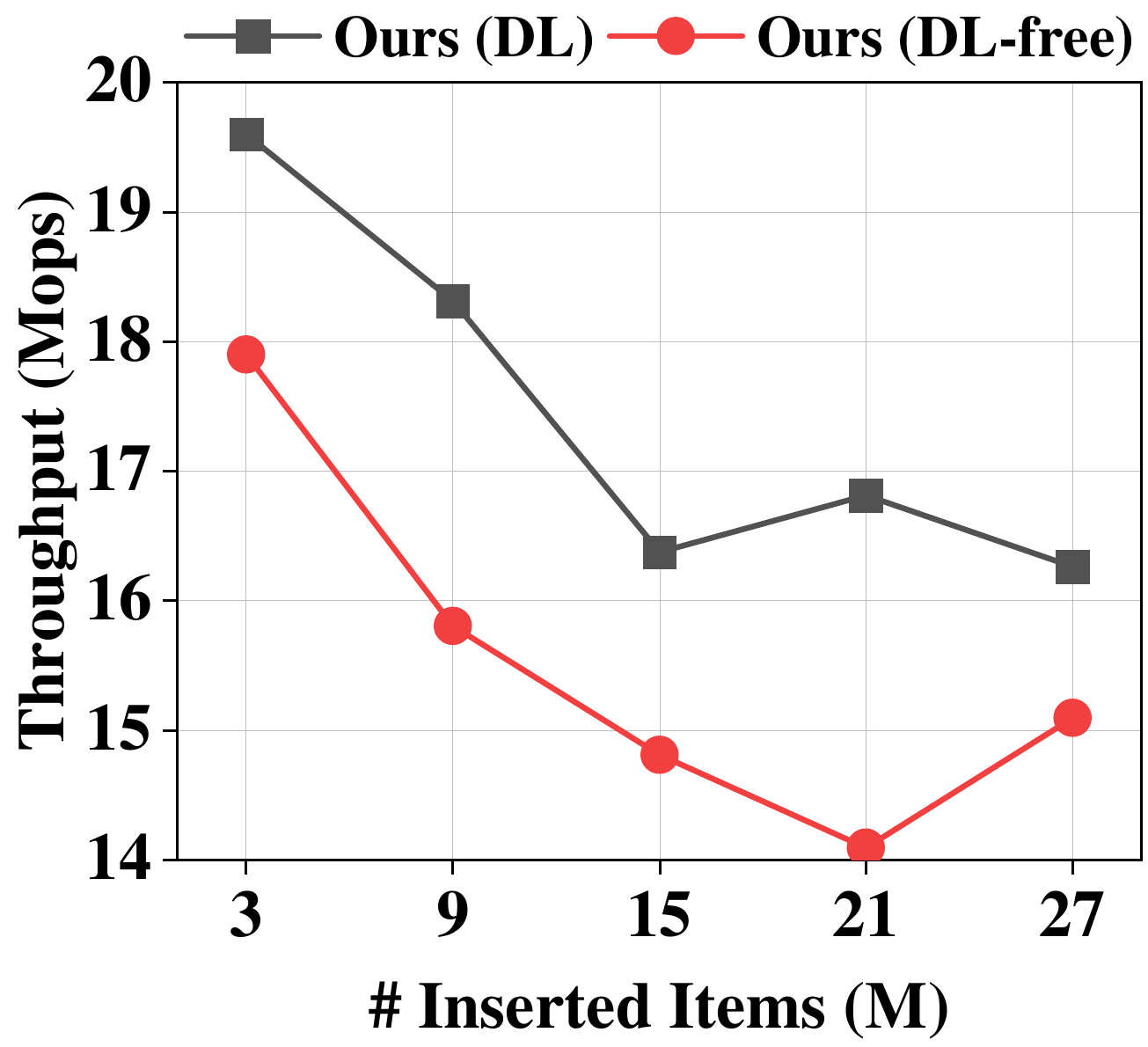}
    \label{DL:qt}
    }
    \subfigure[Memory Usage]{
    \includegraphics[width=4.5cm]{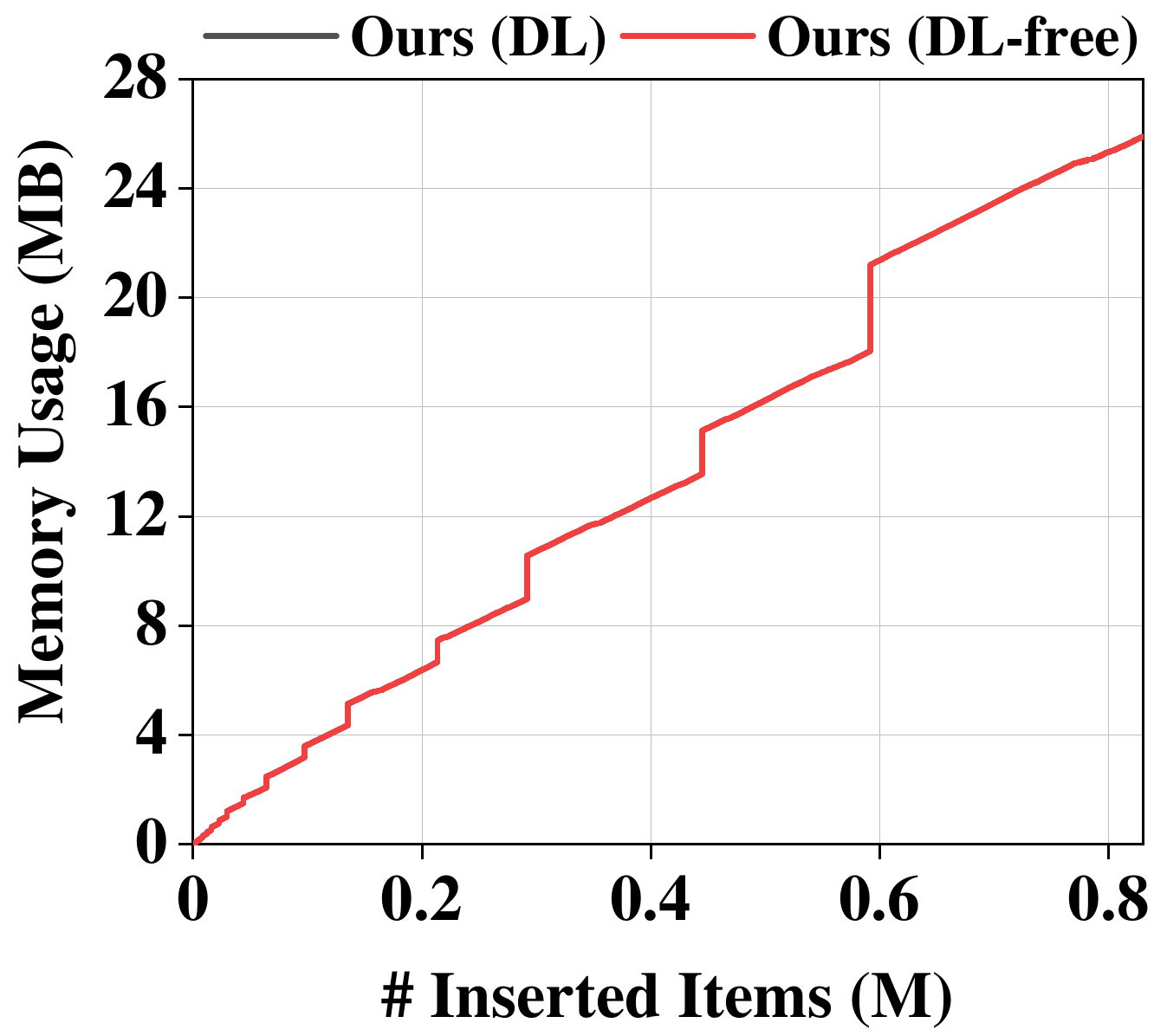}
    \label{DL:mu}
    }
\caption{Ablation experiments: \alg{} with and without DL optimization.}
\label{exp:DL}
\vspace{-0.1in}
\end{figure*}

\bbb{Effects of DL (Figure \ref{exp:DL}): }\textit{The experimental results show that DL indeed further speeds up insertion and querying with almost no additional memory overhead.}
We find that the insertion and query throughput of \alg{} with DL optimization is $1.11\times$ and $1.12\times$ faster than that of \alg{} without DL optimization, respectively.
Also, the memory usage of \alg{} with DL optimization is only about 4KB more than that of \alg{} without DL optimization when all items are inserted.

\subsection{Experiments on Throughput and Memory Usage}
\label{exp:TM}

In this subsection, we evaluate the performance of \alg{} and its competitors in terms of insertion, query, and deletion throughput and memory usage on various graph datasets.

\bbb{Methodology: }1) We insert all edges from the graph dataset into an empty graph structure, and calculate the average insertion throughput;
2) We then query all edges from the graph structure and calculate the average query throughput.
3) We delete edges one by one and calculate the throughput of the process after deletions.
4) We first de-duplicate the datasets to obtain non-duplicated edges, and then insert them into each scheme one by one. 
After each insertion, the physical memory overhead at that moment is output.

\bbb{Insertion throughput (Figure \ref{exp:thp-i}): }The results show that, on the seven datasets, the insertion throughput of \alg{} is $72.17\times$, $32.66\times$, $8.60\times$, and $253.32\times$ faster than that of LiveGraph, Spruce, Sortledton, and WBI on average, respectively.

\begin{figure}
    \centering    \includegraphics[width=0.98\linewidth]{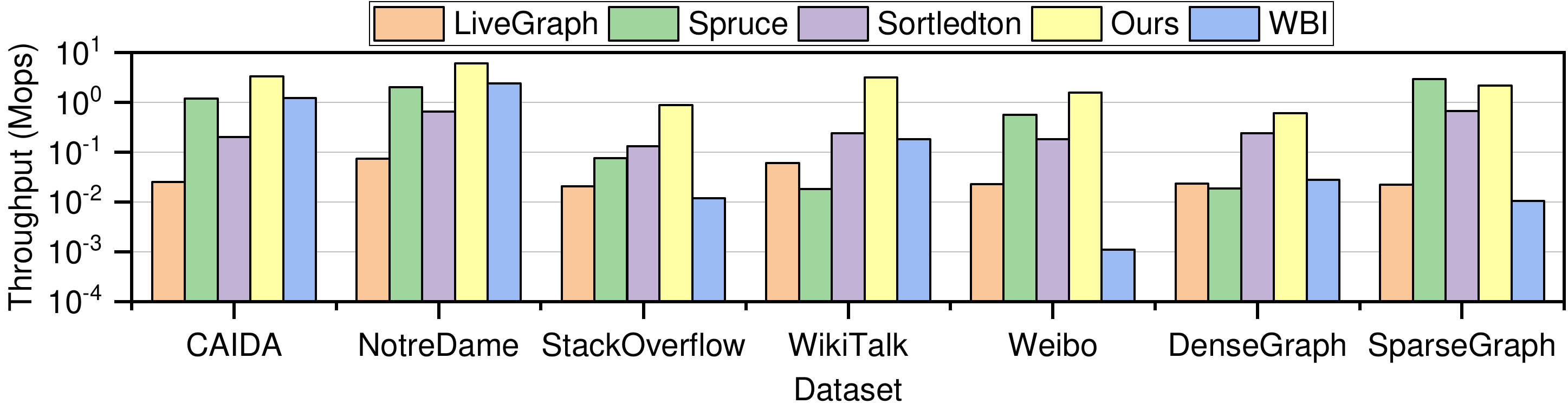}
    \vspace{0.05in}
    \caption{Insertion throughput on different datasets.}
    \label{exp:thp-i}
\end{figure}

\bbb{Query throughput (Figure \ref{exp:thp-q}): }The results show that, on the seven datasets, the query throughput of \alg{} is $14.69\times$, $133.62\times$, $5.34\times$, and $287.48\times$ faster than that of LiveGraph, Spruce, Sortledton, and WBI on average, respectively.

\begin{figure}
    \centering    \includegraphics[width=0.98\linewidth]{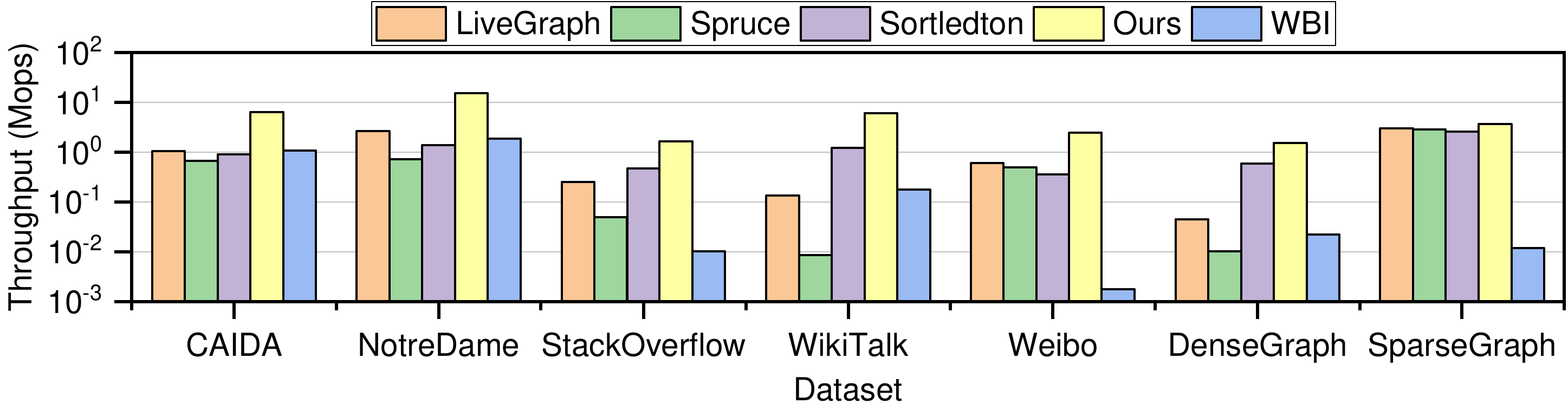}
    \vspace{0.05in}
    \caption{Query throughput on different datasets.}
    \label{exp:thp-q}
    \vspace{-0.1in}
\end{figure}

\bbb{Deletion throughput (Figure \ref{exp:thp-d}): }The results show that, on the seven datasets, the deletion throughput of \alg{} is $85.47\times$, $3.63\times$, $5.01\times$, and $65.55\times$ faster than that of LiveGraph, Spruce, Sortledton, and WBI on average, respectively.

\begin{figure}
    \centering    \includegraphics[width=0.98\linewidth]{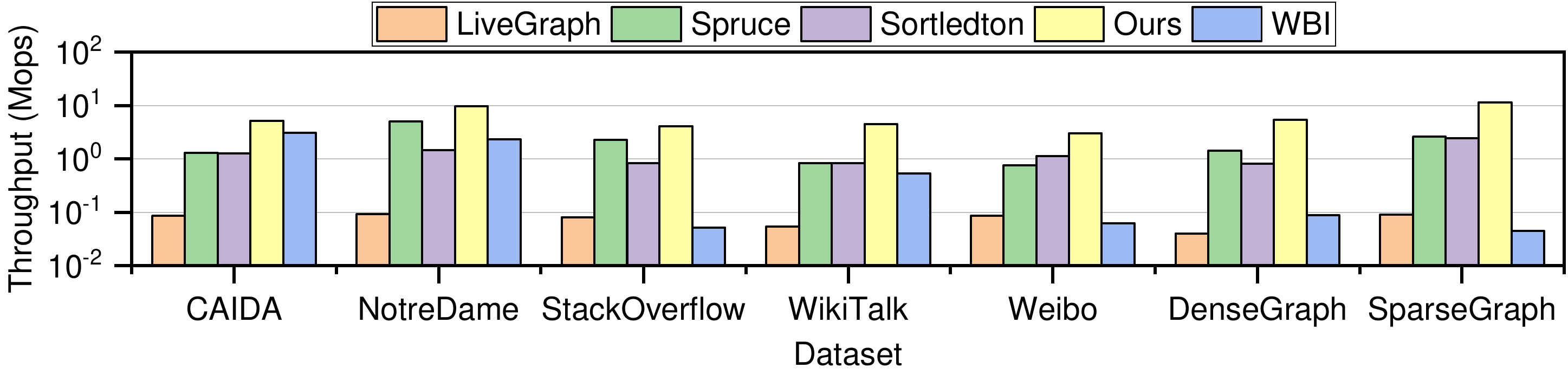}
    \vspace{0.05in}
    \caption{Deletion throughput on different datasets.}
    \label{exp:thp-d}
    \vspace{-0.1in}
\end{figure}

\bbb{Analysis: }1) Thanks to the novel data structure of \alg{}, when inserting or querying an edge, even in the worst case, only 6 buckets in L-CHT and S-CHT, as well as two Denylists, are accessed.
Since the size of the bucket and Denylist is fixed, the upper limit on the number of memory accesses is also fixed and small.
Therefore, no matter how the incoming dataset changes, \alg{} can achieve fast insertion and query.
In contrast, other competitors are designed based on the adjacency list or its variants, so an edge insertion/query operation often requires multiple memory accesses and cannot adapt well to changes in the amount and characteristics of the dataset.
2) For deletions, other schemes simply delete the target when it is found, while \alg{} may involve additional contraction operations.

\begin{figure*}[t!]
\centering
    \subfigure[CAIDA]{
    \includegraphics[width=4.0cm]{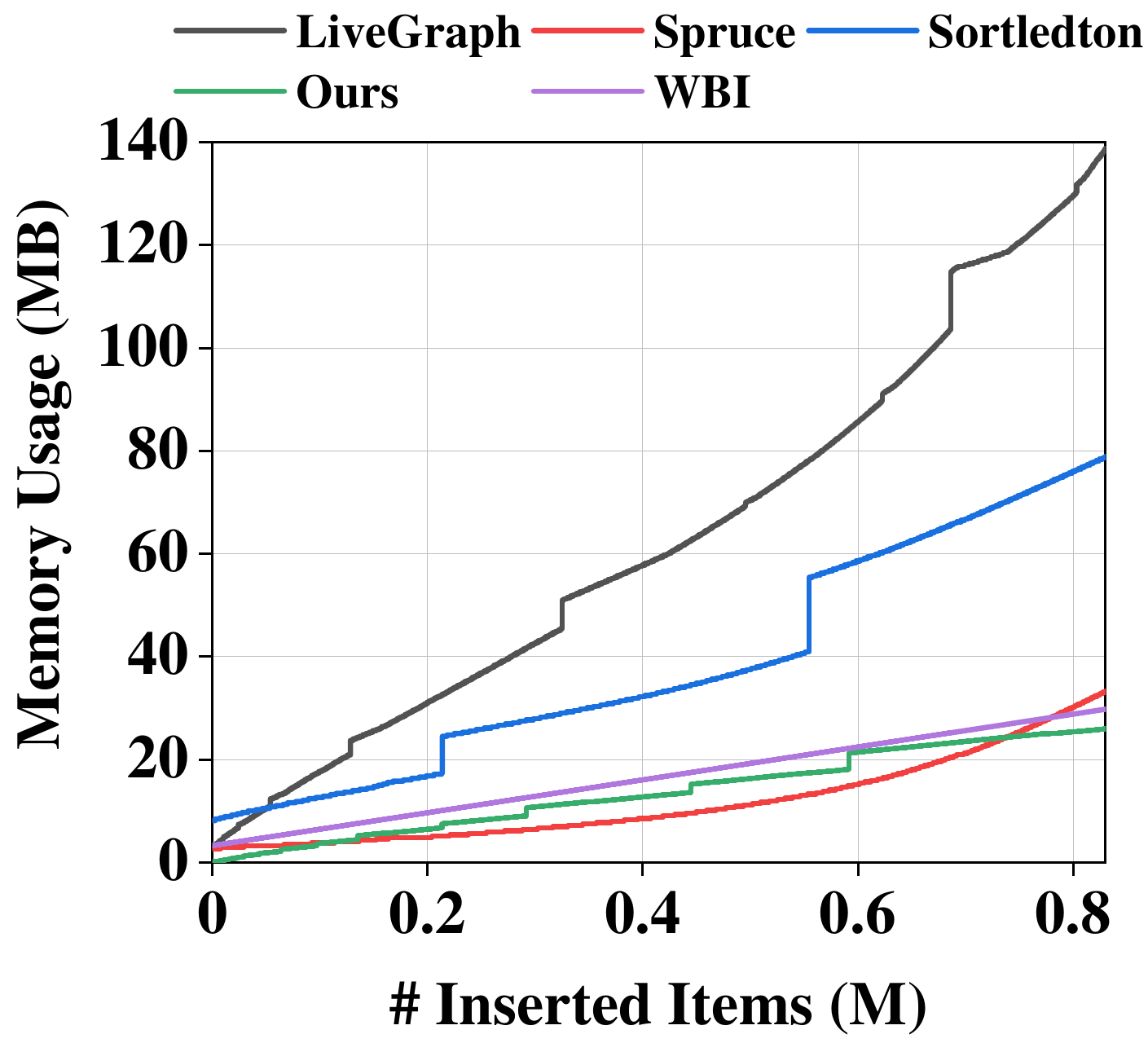}
    \label{mem-c}
    }
    \subfigure[NotreDame]{
    \includegraphics[width=4.0cm]{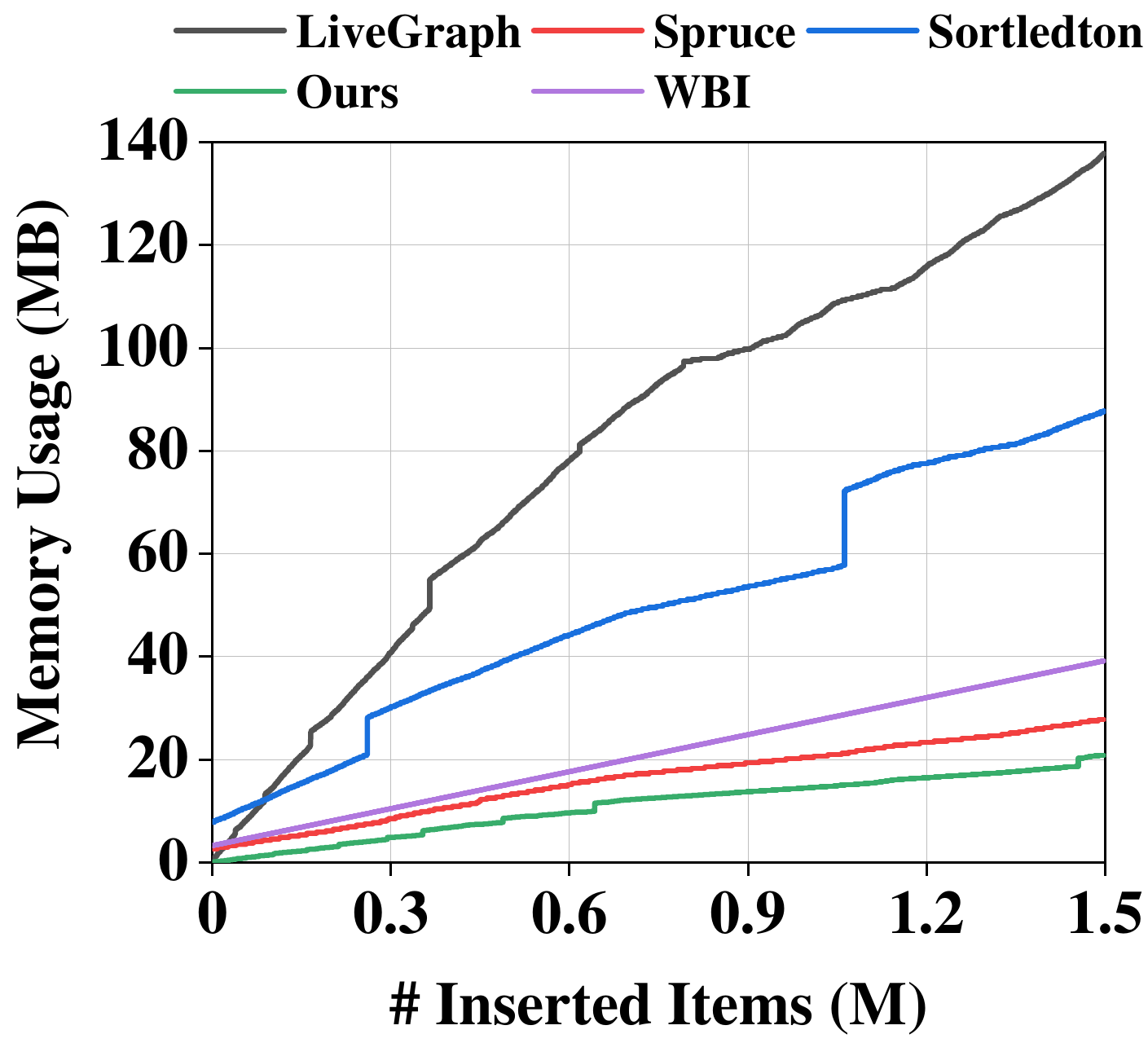}
    \label{mem-n}
    }
    \subfigure[StackOverflow]{
    \includegraphics[width=4.1cm]{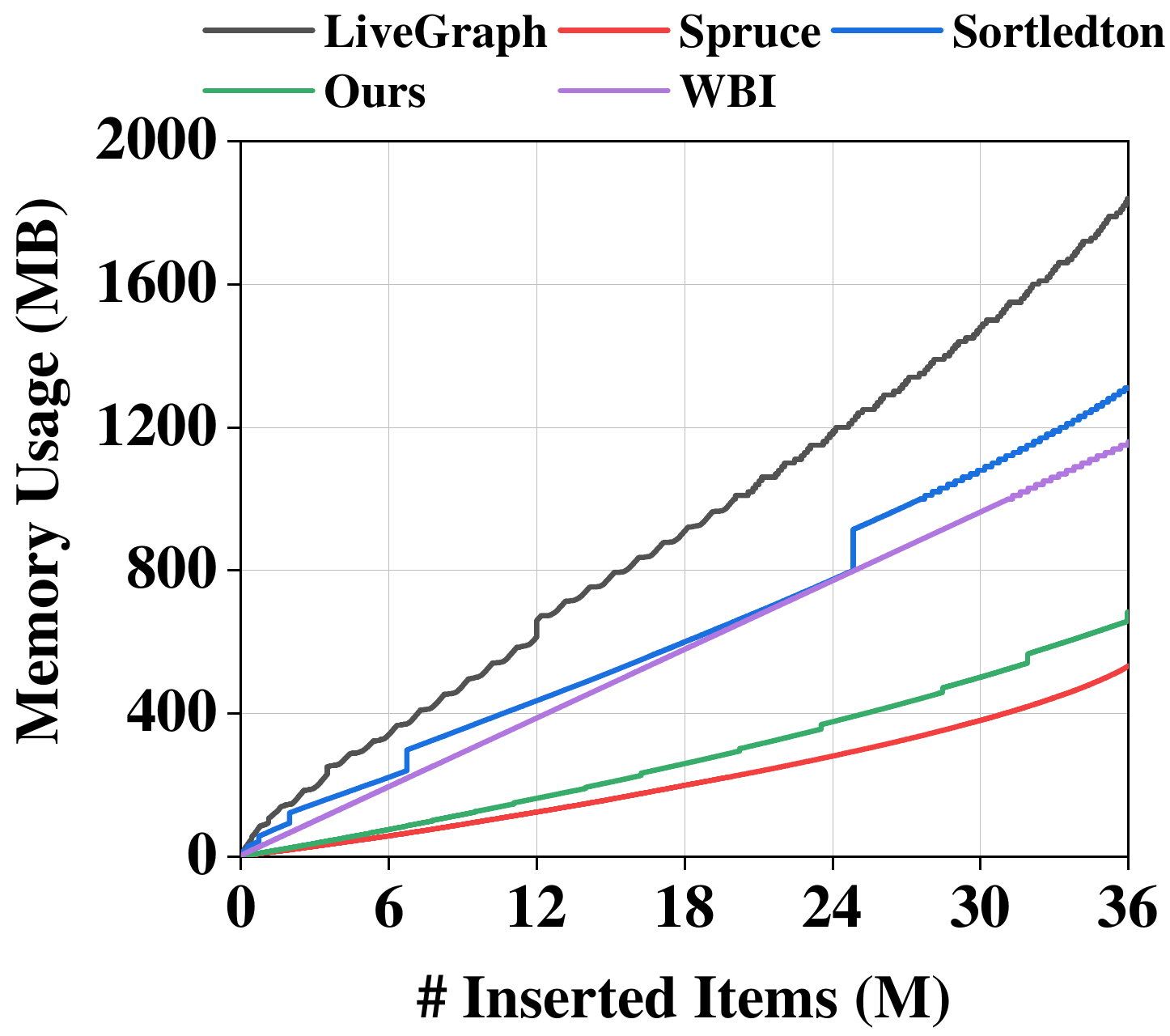}
    \label{mem-s}
    }
    \subfigure[WikiTalk]{
    \includegraphics[width=4.0cm]{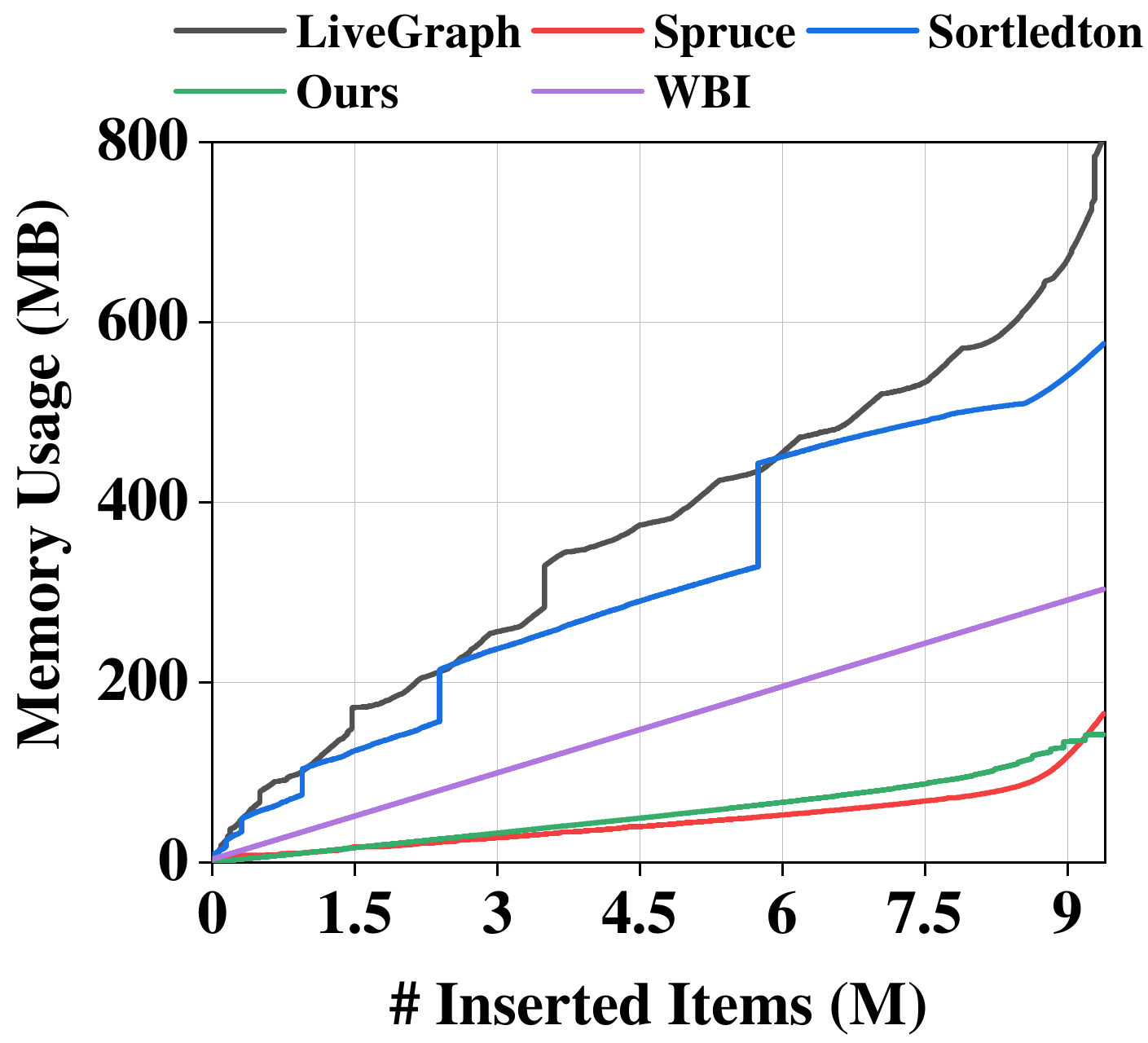}
    \label{mem-wk}
    }
    \subfigure[Weibo]{
    \includegraphics[width=4.4cm]{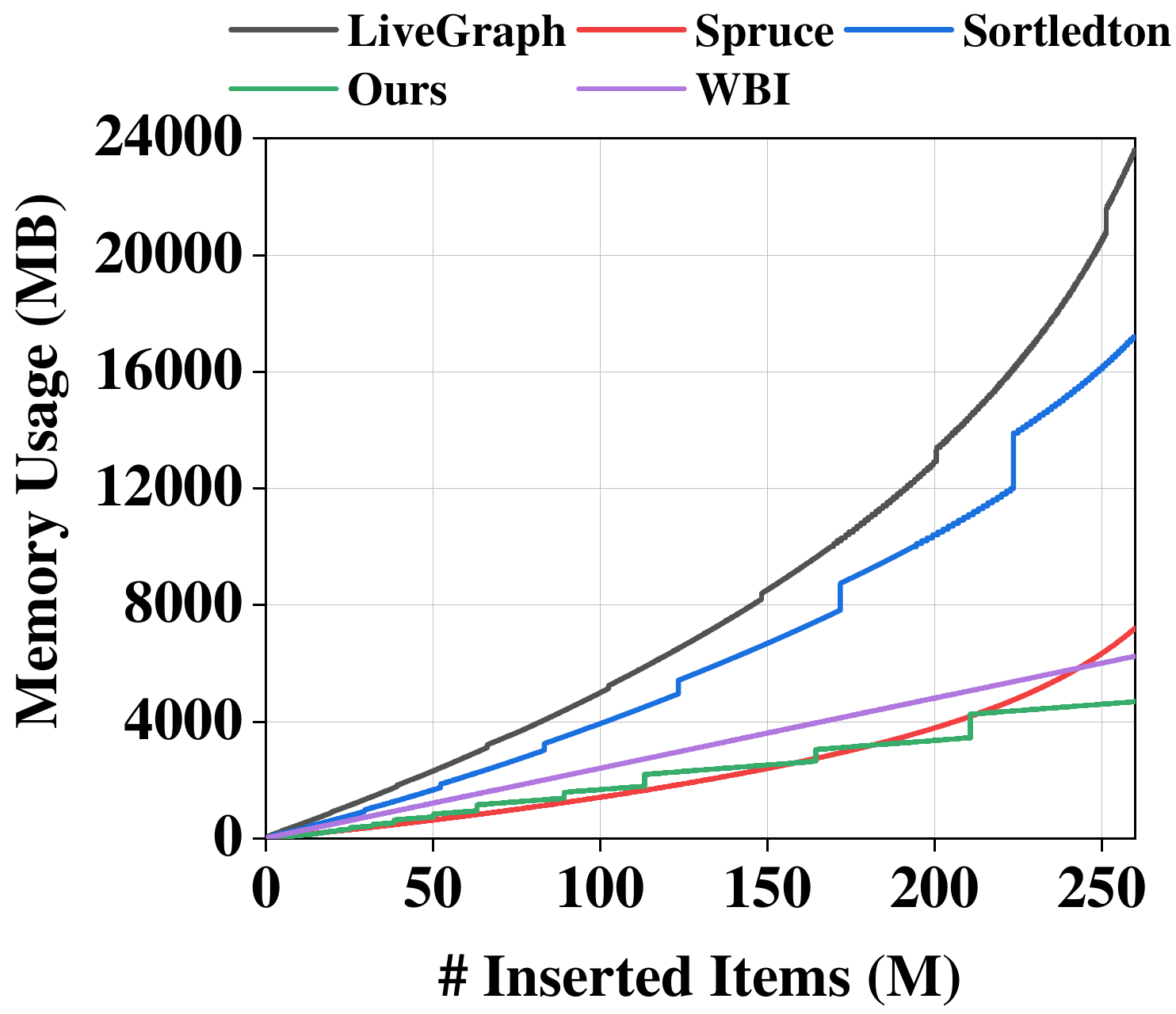}
    \label{mem-wb}
    }
    \subfigure[DenseGraph]{
    \includegraphics[width=4.4cm]{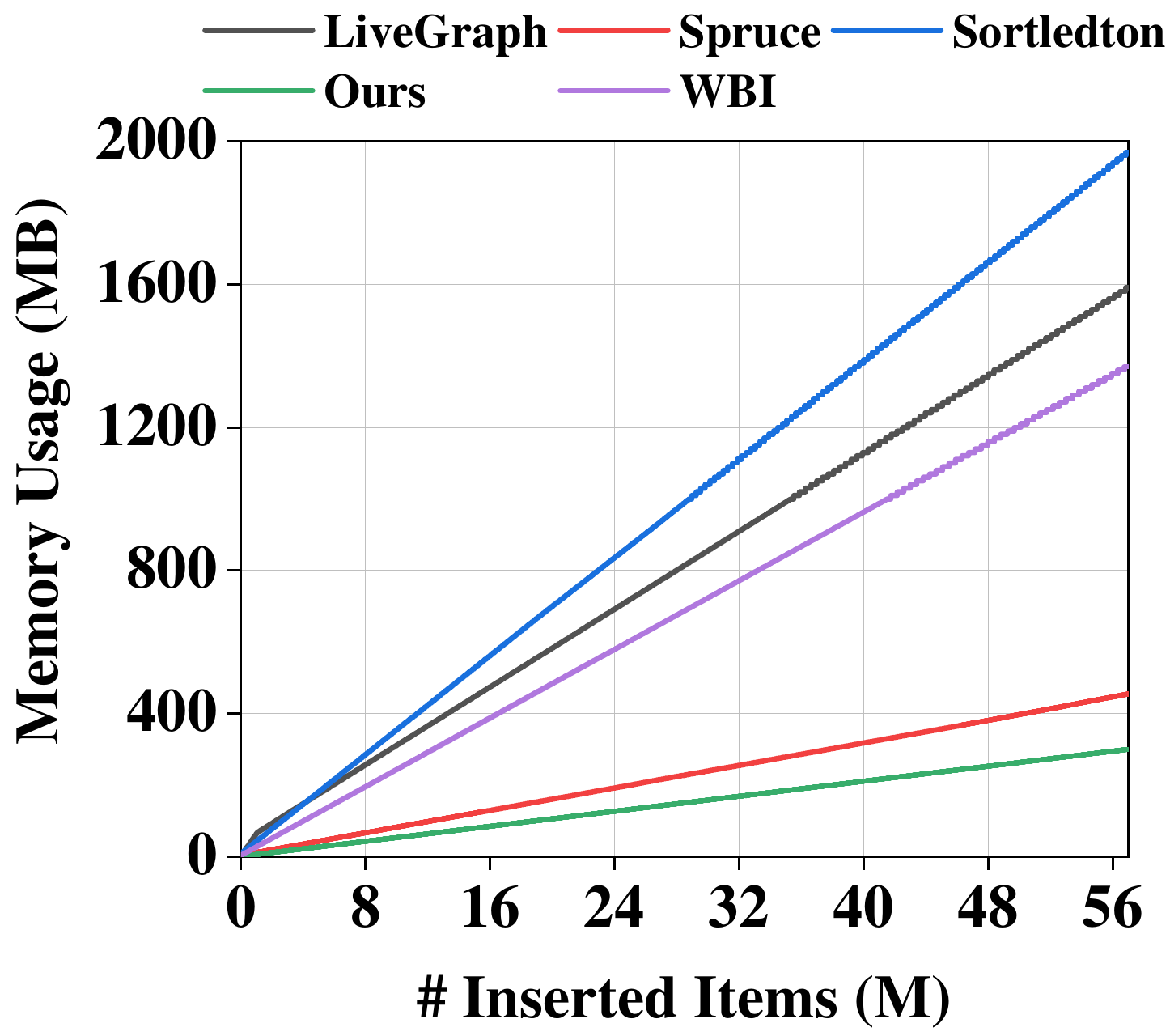}
    \label{mem-de}
    }
    \subfigure[SparseGraph]{
    \includegraphics[width=4.4cm]{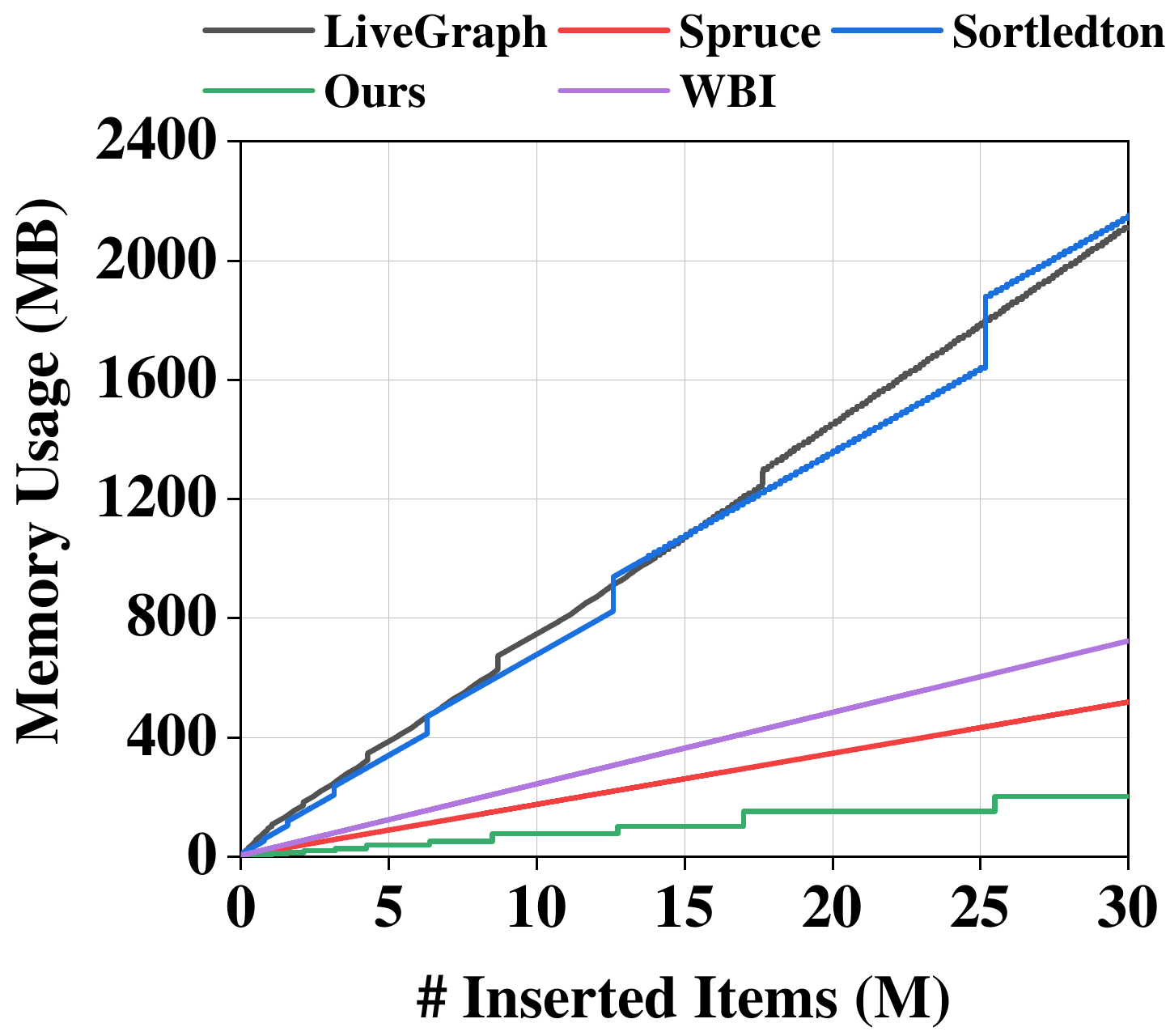}
    \label{mem-sp}
    }
\caption{Memory usage on different datasets.}
\label{mem}
\vspace{-0.1in}
\end{figure*}

\bbb{Memory Usage (Figure \ref{mem-c}-\ref{mem-sp}): }The results show that, on the seven datasets, the memory usage of \alg{} when all item insertions are completed is $5.92\times$, $1.47\times$, $4.89\times$, and $2.34\times$ less than that of LiveGraph, Spruce, Sortledton, and WBI on average, respectively.

\bbb{Analysis: }\alg{} is a customized design based on CHT, so it does not need to store a large number of pointers like the schemes based on adjacency lists, which significantly reduces space overhead.
In addition, since CHT has a high loading rate, \alg{} achieves a high loading rate and minimizes space waste.

\subsection{Experiments on Graph Analytics Tasks}
\label{exp:Analytics}

In this subsection, we evaluate the performance of \alg{} and its competitors in terms of running time on the graph datasets in Table \ref{table:datasets} through the following typical graph analytics tasks: Breadth-First Search (BFS), Single-Source Shortest Paths (SSSP), Triangle Counting (TC), Connected Components (CC), PageRank (PR), Betweenness Centrality (BC), and Local Clustering Coefficient (LCC).
Note that some competitors did not complete the experiments within the given time, so their results are not shown in the provided figures.

\subsubsection{\textbf{Breadth-First Search}}~
\label{task:BFS}

\bbb{Methodology: }We first insert all the edges of the entire dataset. 
Then, we select a specific number of nodes with the largest total degree ($i.e.$, the sum of out-degree and in-degree, the same below), and perform a BFS on these nodes, returning each node and the number of nodes obtained in the order of BFS traversal. 
Finally, we calculate the average time taken for these BFS tasks.

\bbb{Results (Figure \ref{exp:bfs}): }We find that, on the seven datasets (two for WBI), the running time of \alg{} on BFS is $2.34\times$, $0.73\times$, $19.83\times$, and $504.81\times$ faster than that of LiveGraph, Spruce, Sortledton, and WBI on average, respectively.

\begin{figure}
    \centering    \includegraphics[width=0.98\linewidth]{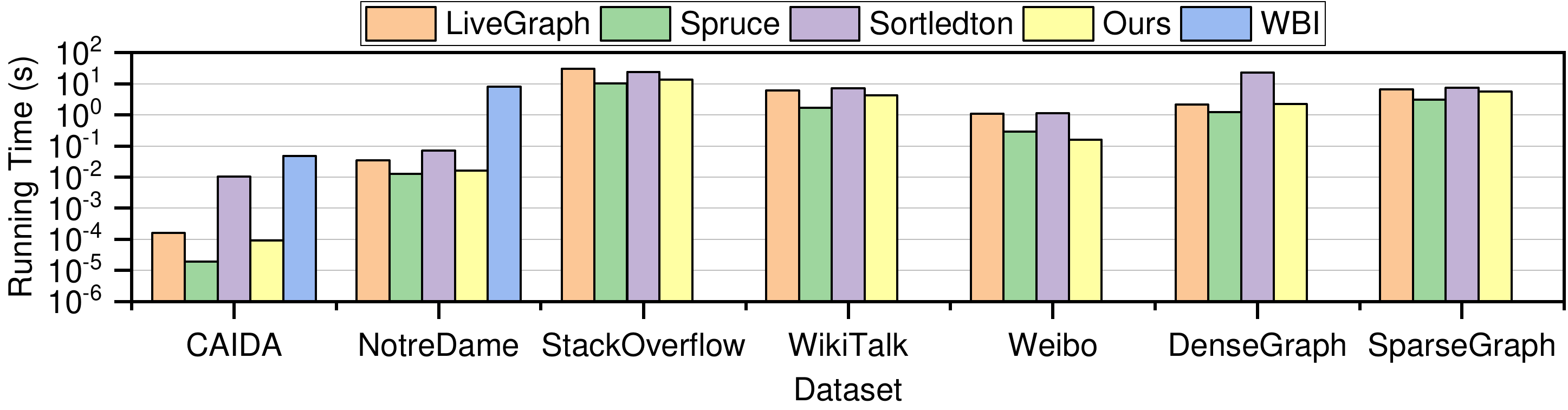}
    \vspace{0.05in}
    \caption{Running time of BFS on different datasets.}
    \label{exp:bfs}
    \vspace{-0.1in}
\end{figure}

\bbb{Analysis: }The most frequently used function of each scheme in this task is its successor query function.
The structure of \alg{} is based on hash tables, so it has good spatial locality. 
Therefore, during the process of querying and traversing the hash tables, the algorithm can achieve excellent spatial locality, which greatly increases the cache hit rate.
Most other adjacency list-based schemes need to store data in different memory addresses and then use pointers to link them.
When querying for successors, the time and space locality is poor and the cache hit rate is low, requiring frequent memory access, which reduces query efficiency.
It is worth noting that WBI not only has the above shortcomings, but also needs to access many other redundant edges when querying successors, so it performs the worst.
The advantage of Spruce may be that its end nodes for finding neighbors can be approximately regarded as being stored more continuously than the other 3.

\subsubsection{\textbf{Single-Source Shortest Paths}}~
\label{task:SSSP}

\bbb{Methodology: }We first insert all the edges of the entire dataset.
Then, we select a specific number of nodes with the largest total degree to extract subgraphs, and select the 10 nodes with the largest total degree among these nodes.
Note that this refers to the 10 nodes with the largest total degree on the original graphs, not on the subgraphs.
After that, we use these 10 nodes as sources to perform Dijkstra algorithm \cite{DijkstraA} 10 times and calculate the average time. 
%

\bbb{Results (Figure \ref{exp:sssp}): }We find that, on the seven datasets (six for Spruce \& WBI), the running time of \alg{} on SSSP is $43.64\times$, $168.45\times$, $1.62\times$, and $278.0\times$ faster than that of LiveGraph, Spruce, Sortledton, and WBI on average, respectively.

\begin{figure}
    \centering    \includegraphics[width=0.98\linewidth]{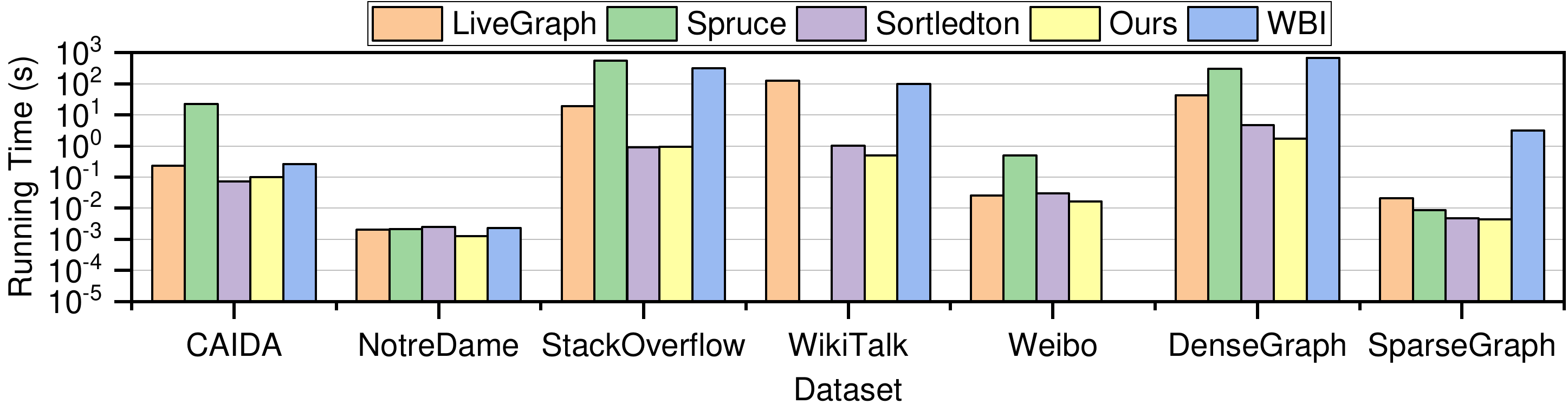}
    \vspace{0.05in}
    \caption{Running time of SSSP on different datasets.}
    \label{exp:sssp}
    \vspace{-0.1in}
\end{figure}

\bbb{Analysis: }The most frequently used function of each scheme in this task is edge query function.
As described in the analysis in $\S$~\ref{exp:TM}, \alg{} has a huge advantage over other adjacency list-based schemes in edge query, so \alg{} achieves the best performance in the SSSP task.

\subsubsection{\textbf{Triangle Counting}}~

\bbb{Methodology: }TC means given a node, return the number of triangles in the graph that contain that node.
First, we perform successor queries to find all 2-hop successors of the node. 
Then, we enumerate all possible edges $\left \langle \textit{2-hop} \ \textit{successor}, \textit{node} \right \rangle$ composed of the node's 2-hop successors and the node itself to perform edge queries. 
Finally, the number of successful queries is the results of TC.

\bbb{Results (Figure \ref{exp:tc}): }We find that, on the seven datasets (five for WBI), the running time of \alg{} on TC is $8.23\times$, $21.33\times$, $1.86\times$, and $3015.11\times$ faster than that of LiveGraph, Spruce, Sortledton, and WBI on average, respectively.

\begin{figure}
    \centering    \includegraphics[width=0.98\linewidth]{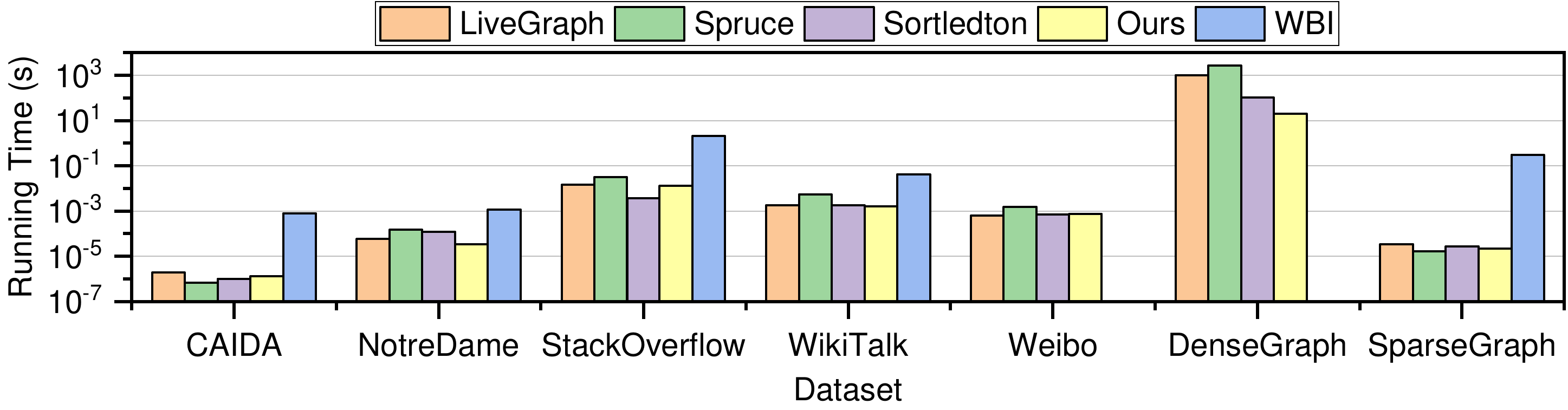}
    \vspace{0.05in}
    \caption{Running time of TC on different datasets.}
    \label{exp:tc}
\end{figure}

\bbb{Analysis: }The edge query and successor query functions are the most frequently used functions in each scheme in this task. 
As analyzed in $\S$~\ref{task:BFS} and $\S$~\ref{task:SSSP}, \alg{} can achieve good performance in these two functions, so it also performs very well in this task.

\subsubsection{\textbf{Connected Components}}~
\label{task:CC}

\bbb{Methodology: }All edges for the entire dataset are inserted.
We first select a specific number of nodes with the largest total degree to extract subgraphs, and then insert the subgraphs into each scheme. \textit{Note that the above steps also apply to the last 3 tasks.}
After that, we run the Tarjan algorithm \cite{tarjan1985efficient} on the subgraphs using each scheme and return the connected components and their number.

\bbb{Results (Figure \ref{exp:cc}): }We find that, on the seven datasets, the running time of \alg{} on CC is $2.11\times$, $1.07\times$, $9.91\times$, and $39.7\times$ faster than that of LiveGraph, Spruce, Sortledton, and WBI on average, respectively.

\begin{figure}
    \centering    \includegraphics[width=0.98\linewidth]{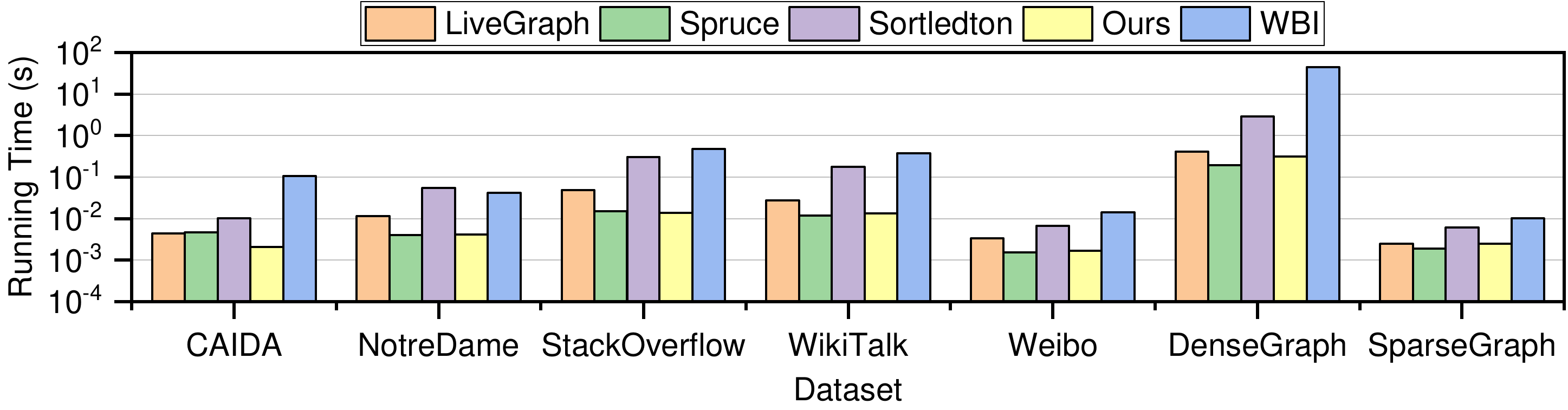}
    \vspace{0.05in}
    \caption{Running time of CC on different datasets.}
    \label{exp:cc}
    \vspace{-0.1in}
\end{figure}

\bbb{Analysis: }The most frequently used function of each scheme in this task is its successor query function.
As analyzed in $\S$~\ref{task:BFS}, \alg{} 
performs well in this function, so it has an advantage over other schemes in this task.

\subsubsection{\textbf{PageRank}}~

\bbb{Methodology: }The initial steps are the same as those in $\S$~\ref{task:CC}.
Then, we use the successor query function of each scheme to assist in constructing the matrix required to solve the PageRank (PR), and iterate 100 times on the matrix to find the PR of each node on the subgraphs.

\bbb{Results (Figure \ref{exp:pr}): }We find that, on the seven datasets, the running time of \alg{} on PR is $2.16\times$, $1.03\times$, $2.62\times$, and $2.87\times$ faster than that of LiveGraph, Spruce, Sortledton, and WBI on average, respectively.

\begin{figure}
    \centering    \includegraphics[width=0.98\linewidth]{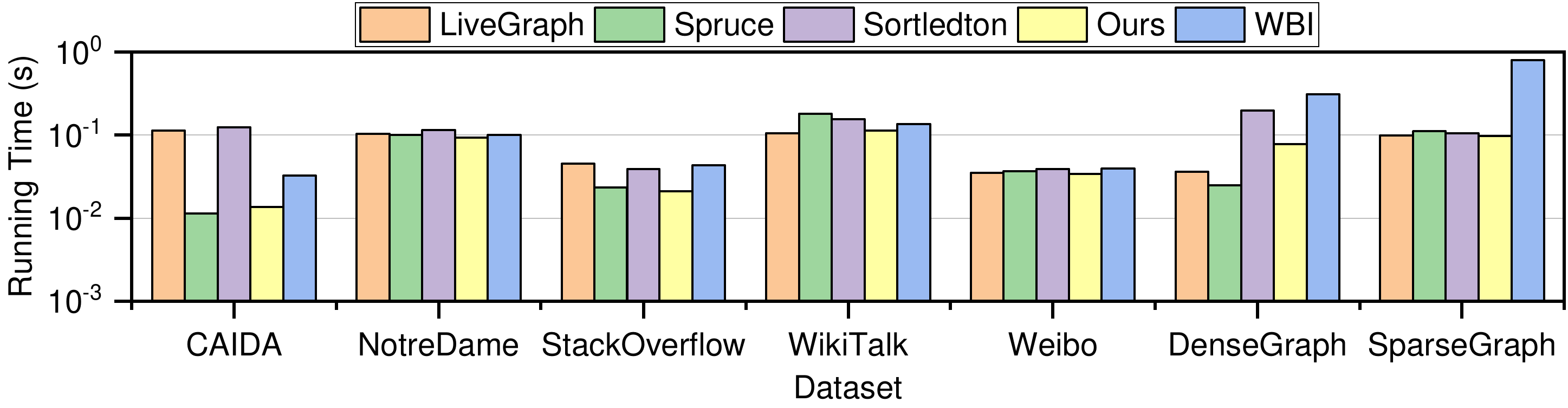}
    \vspace{0.05in}
    \caption{Running time of PR on different datasets.}
    \label{exp:pr}
\end{figure}

\bbb{Analysis: }Each scheme in this task frequently uses the successor query function to construct the matrix required to calculate PR.
As analyzed in $\S$~\ref{task:BFS}, \alg{} 
performs well in successor query, so it also shows advantages over other schemes in this task.

\subsubsection{\textbf{Betweenness Centrality}}~

\bbb{Methodology: }The initial steps are the same as those in $\S$~\ref{task:CC}.
Then, we run the Brandes algorithm \cite{brandes2001faster} on the subgraphs using each scheme.

\bbb{Results (Figure \ref{exp:bc}): }We find that, on the seven datasets, the running time of \alg{} on BC is $3.15\times$, $16.17\times$, $7.33\times$, and $5.23\times$ faster than that of LiveGraph, Spruce, Sortledton, and WBI on average, respectively.

\begin{figure}
    \centering    \includegraphics[width=0.98\linewidth]{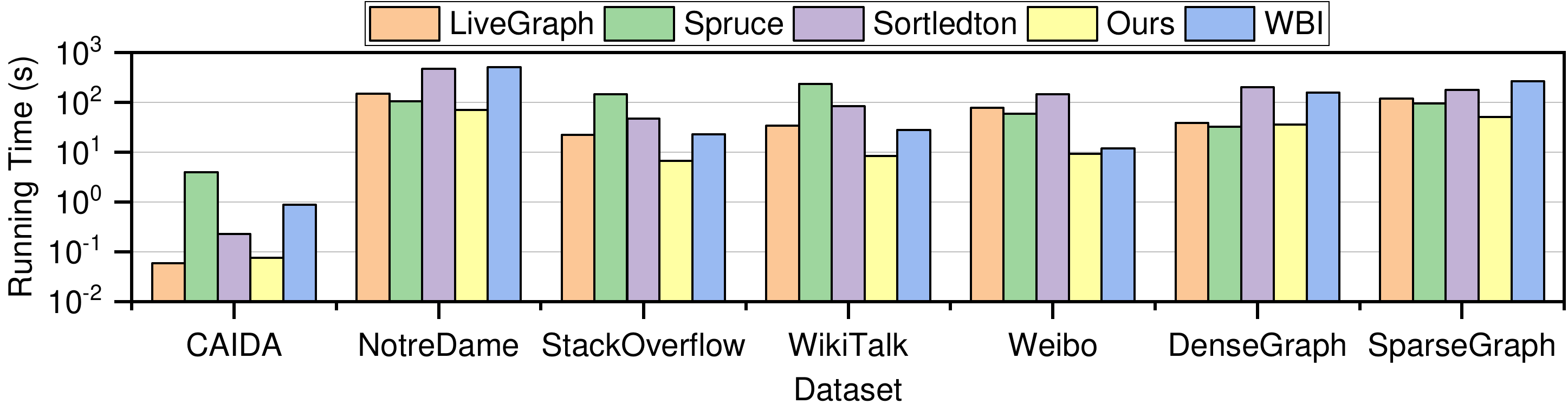}
    \vspace{0.05in}
    \caption{Running time of BC on different datasets.}
    \label{exp:bc}
    \vspace{-0.1in}
\end{figure}

\bbb{Analysis: }Similar to the analysis in $\S$~\ref{task:CC}.

\subsubsection{\textbf{Local Clustering Coefficient}}~

\bbb{Methodology: }The initial steps are the same as those in $\S$~\ref{task:CC}. 
Then, we pre-compute all neighbors of each node and run the Local Clustering Coefficient (LCC) algorithm, which is implemented in \cite{iosup2020ldbc}.

\bbb{Results (Figure \ref{exp:lcc}): }We find that, on the seven datasets (six
for WBI), the running time of \alg{} on LCC is $2.06\times$, $5.80\times$, $3.94\times$, and $4.21\times$ faster than that of LiveGraph, Spruce, Sortledton, and WBI on average, respectively.

\begin{figure}
    \centering    \includegraphics[width=0.98\linewidth]{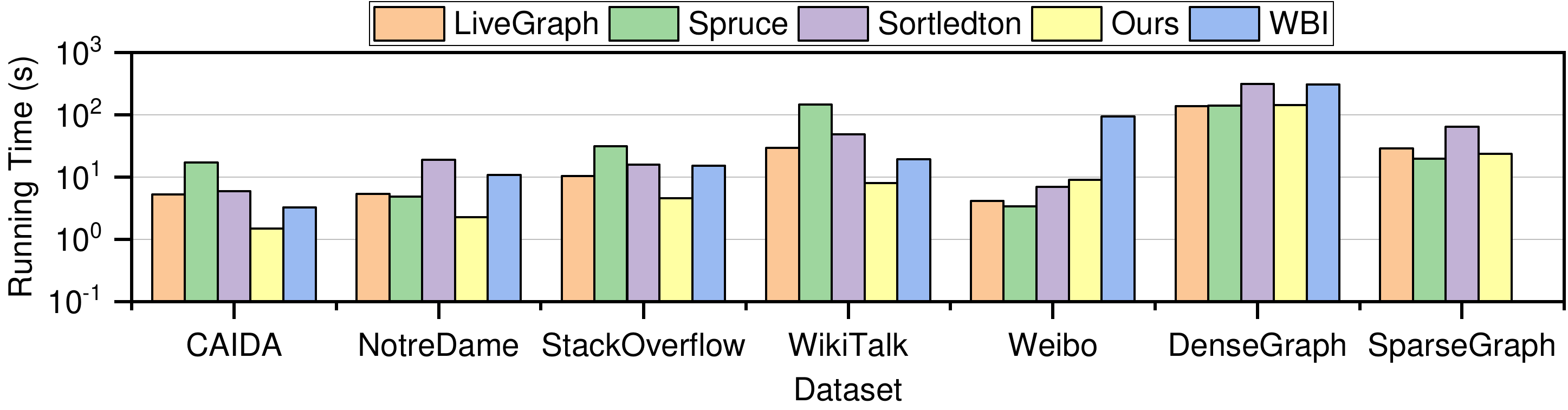}
    \vspace{0.05in}
    \caption{Running time of LCC on different datasets.}
    \label{exp:lcc}
    \vspace{-0.1in}
\end{figure}

\bbb{Analysis: }Similar to the analysis in $\S$~\ref{task:CC}.

\subsection{Redis Implementation}
\label{Redis-I}

\bbb{Methodology: }We utilize Redis Module \cite{redis-m} to register our \alg{} module, adding the data structure of \alg{} to the original Redis. 
This allows Redis to store graphs in addition to supporting the five original data structures.
Specifically, we implement Redis Module API (including \texttt{save\_rdb}, \texttt{load\_rdb}, \texttt{aof\_rewrite} and other interfaces) on top of \alg{} to support Redis persistence operations.
Meanwhile, we also provide extended commands for \alg{} (including \texttt{insert}, \texttt{del}, \texttt{query} and \texttt{getneighbors}).
We compile our interface implementation into a dynamic link library, and simply import \alg{} library with \texttt{--loadmodule} when Redis starts.

\bbb{Setup: }We conduct experiments on the CAIDA and StackOverflow datasets to test the throughput performance of \alg{} on Redis.

\bbb{Results \& Analysis (Figure \ref{exp:redis}): }The results show that the insertion, query, and deletion throughput of \alg{} on Redis is around $0.04\sim0.05$ Mops. 
There is some performance loss compared to those of \alg{} on CPU, which is mainly caused by the Redis system.
We also run Redis benchmark on the server, and the peak throughput of native Redis is only around $0.16$ Mops. 
Considering that \alg{} itself inevitably has overhead, its performance on Redis is completely acceptable.

\begin{figure}
    \centering    \includegraphics[width=0.95\linewidth]{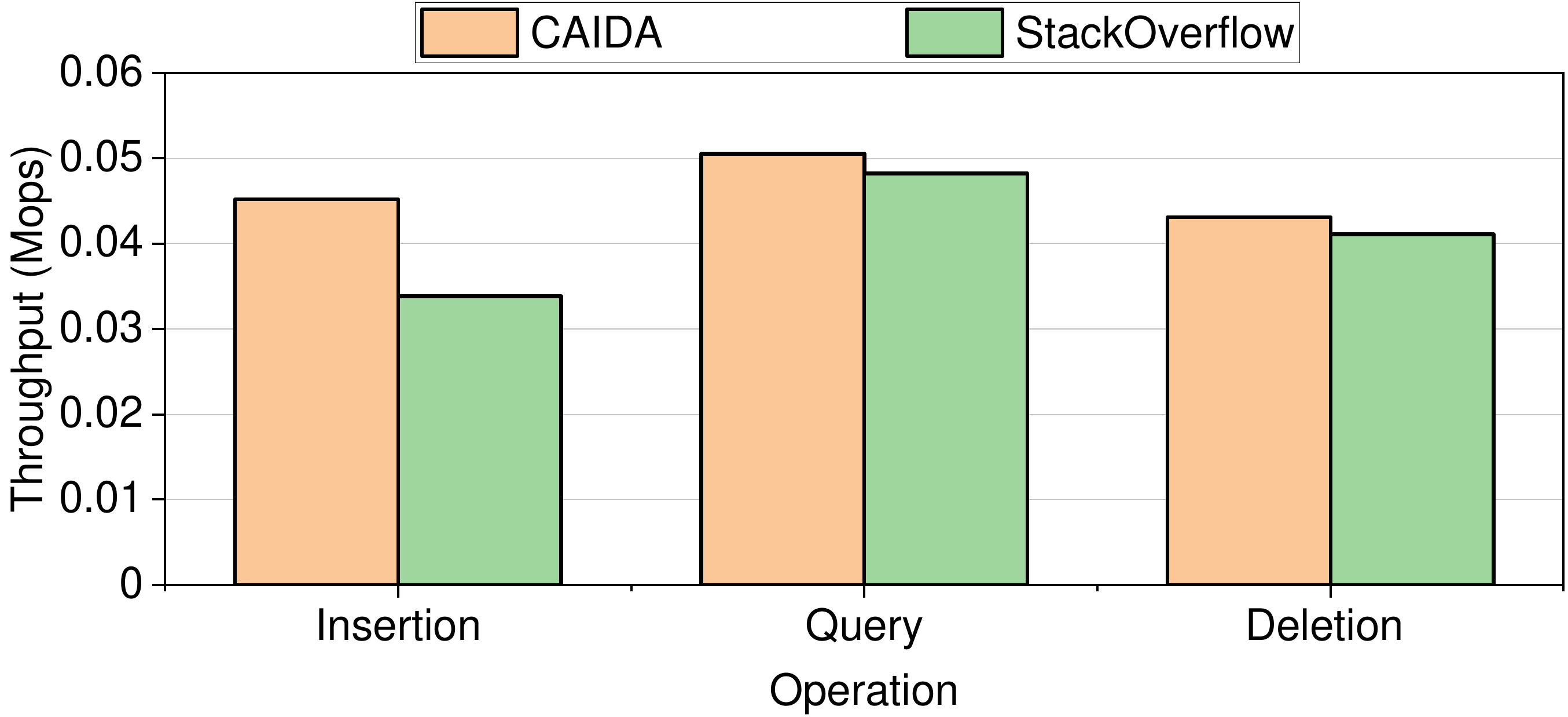}
    \vspace{0.05in}
    \caption{The throughput of \alg{} on Redis.}
    \label{exp:redis}
    \vspace{-0.15in}
\end{figure}

\subsection{Neo4j Implementation}
\label{Neo4j-I}

\bbb{Methodology: }If we want to store an edge $\left \langle u, v \right \rangle$ in Neo4j \cite{Neo4j}, nodes $u$ and $v$ each maintain a adjacency list that stores all the edges associated with that node, so the information about $\left \langle u, v \right \rangle$ is stored in the adjacency lists of both $u$ and $v$.
If we want to query an edge $\left \langle u, v \right \rangle$, we have to find the adjacency list of $u$, and then traverse the list and compare the edges one by one until we find $\left \langle u, v \right \rangle$. 
Obviously, it is inefficient. 
Once the degree of $u$ is high, querying edge $\left \langle u, v \right \rangle$ has to access a large number of unrelated other edges, causing additional redundancy overhead.
To speed up edge queries, we introduce the \alg{} query interface to obtain an edge without traversing the adjacency list of the node.
Since multiple edges (with the same $u$ and $v$ but not the same edge) are allowed in Neo4j, the data structure of \alg{} needs some adjustments for this.
Compared to the weighted version on the CPU, we change the weight field in each S-CHT small slot from a counter that records the number of edges to a linked list consisting of a series of edges with the same nodes $u$ and $v$. 
The linked list is as long as the number of edges corresponding to $\left \langle u, v \right \rangle$ in that small slot.
In this way, the query interface of \alg{} returns an iterator, through which the linked list can be traversed to obtain all the edges between $\left \langle u, v \right \rangle$.

\bbb{Setup: }We deploy the above \alg{} on top of the original Neo4j and evaluate the performance by running time, as shown below.
1) For the insertion experiments, we insert the first 1M edges from the CAIDA dataset into Neo4j.
Whenever an edge is inserted into Neo4j, we also need to insert that edge into the \alg{} structure, which requires a little extra time overhead.
2) For the query experiments, we first deduplicate the 1M edges, and then query the \alg{} structure. 
For comparison, related operations on pure Neo4j do not introduce \alg{}.


\bbb{Results \& Analysis (Figure \ref{exp:Neo4j}): }\textbf{1)} Thanks to the good performance of our data structure, our insertions are very fast and require only a little extra overhead, so our insertion time is almost the same as pure Neo4j.
\textbf{2)} Since the time cost of \alg{}'s query to obtain the iterator of the linked list is $O(1)$, the query speed of the version with \alg{} is very fast.
It can be predicted that the query speed of Neo4j with \alg{} will be improved more significantly as the data scale increases.
In pure Neo4j, many irrelevant/redundant edges must be traversed, and this additional overhead time is not $O(1)$, which ultimately causes the query time of pure Neo4j to be much slower than that with the assistance of \alg{}.

\begin{figure}
    \centering    \includegraphics[width=0.88\linewidth]{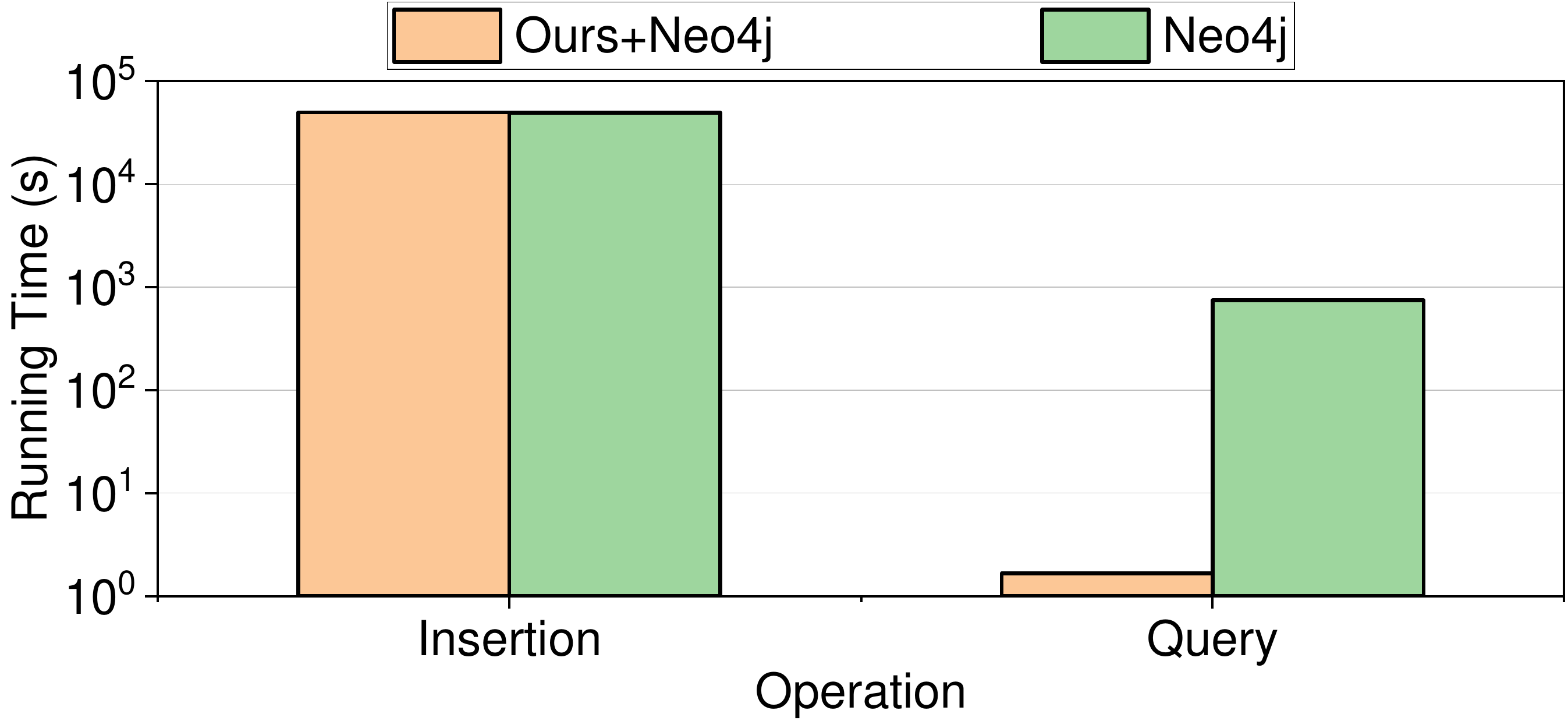}
    \vspace{0.05in}
    \caption{Running time of Neo4j with and without \alg{}.}
    \label{exp:Neo4j}
    \vspace{-0.15in}
\end{figure}

%% file: VLDB/6-con.tex
\section{Conclusion}

In this paper, we propose a novel data structure designed for large-scale dynamic graphs, called \alg{}, which includes two key techniques, \textsc{Transformation} and \textsc{Denylist}. 
Thanks to them, \alg{} can be flexibly resized based on actual operations to achieve memory efficiency while keeping few memory accesses to achieve fast processing speed without any prior knowledge of the upcoming graphs.
Our mathematical analysis theoretically proves that \alg{} is time and space efficient.
Our experimental results show that \alg{} significantly outperforms 4 state-of-the-art schemes.
In particular, compared with Spruce, \alg{} achieves $32.66\times$ faster insertion throughput while reducing memory space by about $32\%$, and $168.45\times$ less running time on the SSSP task.
Finally, we integrate \alg{} in Redis and Neo4j databases to extend its practicality.